\RequirePackage{amsmath}
\documentclass[runningheads]{llncs}
\usepackage[T1]{fontenc}
\usepackage{amssymb,amsfonts}
\usepackage{float}
\usepackage{physics}
\usepackage{pgfplots}
\usepackage[]{mdframed}
\usepackage{bbm}
\usepackage{nicefrac}
\usepackage{enumitem}
\usepackage[c]{esvect}
\usepackage{fullpage}

\usetikzlibrary{circuits.logic.US}
\usetikzlibrary{positioning,shapes.callouts}
\usetikzlibrary{arrows,decorations.pathreplacing}
\usetikzlibrary{decorations.markings}
\usetikzlibrary{decorations.pathmorphing,patterns}
\usetikzlibrary{backgrounds,fit,decorations.pathreplacing,calc}
\usetikzlibrary{overlay-beamer-styles}
\usetikzlibrary{shapes.geometric}
\usetikzlibrary{shapes}

\newcommand{\id}{{\mathbf{1}}}

\newcommand{\Prob}{\operatorname{Pr}}

\newcommand{\kfoldV}{V}
\newcommand{\kfoldv}{v}
\newcommand{\kprodvset}{\mathcal{V}}
\newcommand{\kprodvsetsize}{\abs{\mathcal{V}}}
\newcommand{\kfoldspezv}{\bar{v}}

\newcommand{\nfoldX}{X^{(n)}}
\newcommand{\nfoldx}{x^{(n)}}
\newcommand{\nprodxset}{\mathcal{X}^{n}}

\newcommand{\nfoldY}{Y^{(n)}}
\newcommand{\nfoldy}{y^{(n)}}
\newcommand{\nprodyset}{\mathcal{Y}^{n}}

\newcommand{\nfoldZ}{Z^{(n)}}
\newcommand{\nfoldz}{z^{(n)}}
\newcommand{\nprodzset}{\mathcal{Z}^{n}}
\newcommand{\nprodzsetsize}{\abs{\mathcal{Z}}^{n}}

\newcommand{\nfoldQ}{Q^{(n)}}
\newcommand{\nfoldq}{q^{(n)}}

\newcommand{\lfoldm}{m}
\newcommand{\lprodmset}{\mathcal{M}}
\newcommand{\lprodmsetsize}{\abs{\mathcal{M}}}

\newcommand{\round}{\lceil \nicefrac{\nprodzsetsize}{\kprodvsetsize} \rceil}

\newtheorem{protocol}{Protocol}
\newtheorem{condition}{Condition}

\pgfplotsset{compat=1.18}

\begin{document}

\title{Security for Adversarial Wiretap Channels}

\date{\today}

\author{Esther H\"anggi\inst{1} \and Iy\'an M\'endez Veiga\inst{1,2} \and Ligong Wang\inst{1}}

\institute{Lucerne School of Computer Science and Information Technology, Lucerne University of Applied Sciences and Arts, Rotkreuz, Switzerland
 \and Institute for Theoretical Physics, ETH Zurich, Zurich, Switzerland}

\maketitle

\begin{abstract}
We consider the wiretap channel, where the individual channel uses have memory or are influenced by an adversary. We analyze the explicit and computationally efficient construction of information-theoretically secure coding schemes which use the inverse of an extractor and an error-correcting code. These schemes are known to achieve secrecy capacity on a large class of memoryless wiretap channels. We show that this also holds for certain channel types with memory. In particular, they can achieve secrecy capacity on channels where an adversary can pick a sequence of ``states'' governing the channel's behavior, as long as, given every possible state, the channel is strongly symmetric.
\end{abstract}

\section{The Wiretap Channel}\label{sec:overview}

The goal of the wiretap channel is for two honest parties, a sender and a receiver, connected by a (possibly noisy) communication channel to communicate secretly. The eavesdropper obtains a copy of all the messages sent over the channel, however, in a `noisier' version. 

This classic problem from information theory dates back to the $1970$'s
and was first studied by Wyner~\cite{wyner} and Csisz{\'a}r and K\"orner~\cite{csikor}. The security of the scheme \emph{only} relies on the noise in the communication channel and does not rely on any computational assumption. With the advance of quantum computers and their ability to break~\cite{Shor} {RSA}~\cite{rsa} or Diffie-Hellman~\cite{dh}, information-theoretically secure schemes such as the wiretap channel have seen renewed interest in recent years.

\begin{figure}[h!]
\centering
\begin{tikzpicture}	[line width=1pt]
    \path[clip] (-2,-1.0) rectangle (11.65,2.5);
    \node[align=center, anchor=center, rotate=90] at (-1,1.0) {\small{SENDER}};
    \node[align=center, anchor=center, rotate=-90] at (10.5,1.0) {\small{RECEIVER}};
    \node[align=center, anchor=center] at (0.5,0.5) {message $m$};
    \draw[->, line width=0.2mm, color=black] (0.5,1.0) -- (0.5,2) node[pos=0.5, anchor = west] {$f$};
    \node[align=center, anchor=south] at (-0.3,2.0) {input $\nfoldx{=}f(m)$};
    \draw[-,line width=1pt,double distance=10pt] (1,2.2) -- (8,2.2) node[pos=0.5, anchor = center] {$ChR$};
    \node[align=center, anchor=south, minimum width=1cm] at (9.7,2.0) {output $\nfoldy{=}ChR(\nfoldx)$};
    \draw[->, line width=0.2mm, color=black] (8.5,2.0) -- (8.5,1) node[pos=0.5, anchor = west] {$\bar{f}$};
        \node[align=center, anchor=south] at (8.5,0.5) {message $m'{=}\bar{f}^{-1}(\nfoldy)$};
            \draw[-,line width=1pt,double distance=10pt] (1,2.2) -- (4,0) node[pos=0.5, anchor = center, rotate=-35 ] {$ChA$};
    \node[align=center, anchor=south, minimum width=1cm] at (4.5,-1) {\small{EAVESDROPPER}};
    \node[align=center, anchor=south] at (5.5,-0.5) {output $\nfoldz{=}ChA(\nfoldx)$};
    \end{tikzpicture}
       \caption{\label{fig:wiretap-concept}Concept of the wiretap channel.}
\end{figure}
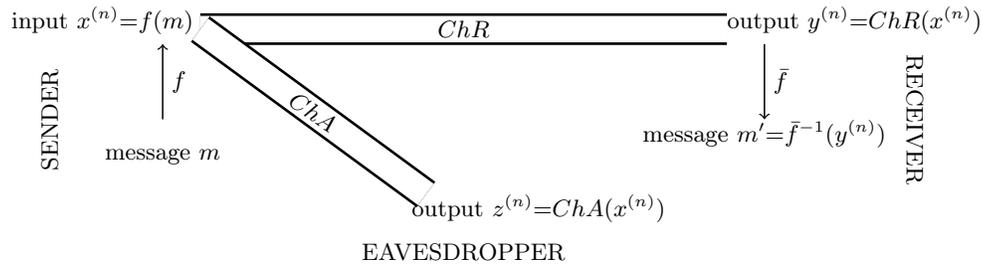

A wiretap scheme should reach two properties: first, the receiver should obtain the correct message, called \emph{correctness} (or \emph{low probability of decoding error}). This property relates to the channel connecting the honest parties $ChR$. The second property, called \emph{secrecy}, means the eavesdropper should not learn the message from the output of the adversarial channel $ChA$. 

The best asymptotic rate at which reliable and secret communication is possible is called \emph{secrecy capacity}. 
This was first studied by Wyner~\cite{wyner}, and  a general formula for the secrecy capacity of a discrete memoryless wiretap channel was given by Csisz\'ar and K\"orner \cite{csikor}: 
\begin{align*}
C_\text{sec}&=\max_{P_{VX}}\left(  I(V;Y) - I(V;Z)\right)\ .
\end{align*}
For the case when both channels are binary symmetric channels this amounts to 
$C_\text{sec}= h(p_A)-h(p_R)$, 
where $p_A$ and $p_R$ are the respective error probabilities of the adversarial and the receiver channel. 

A large body of research has focused on finding the secrecy capacity for the case where the same memoryless channel is applied repeatedly. The main focus was on the \emph{existence} of a wiretap scheme and considering its rate in an \emph{asymptotic setting}. The secrecy capacity has been determined e.g.\ for the case where the individual adversarial channels are a degraded version of the receiver's channel~\cite{wyner} and for Gaussian channels~\cite{gaussian}.

Depending on the context, the precise secrecy and correctness requirements differ. 
In \emph{information theory} the message is typically assumed to be chosen uniformly at random and secrecy and correctness are defined accordingly. \emph{Strong secrecy} is quantified by the \emph{mutual information} between the sender's input and the eavesdropper's output and this quantity is required to vanish asymptotically with a large number of channel uses. Secrecy is called \emph{weak} when the mutual information per channel use vanishes. 

In \emph{cryptography}, the \emph{worst-case} is considered instead of the \emph{average-case}. Correctness and secrecy are expected to hold when the distribution of messages is arbitrary. This prevents the adversary from retrieving any partial information or to distinguish between only two possible messages, e.g. corresponding to a ``yes'' and ``no''. 
This naturally leads to quantifying security by \emph{semantic security} or \emph{distinguishing security}~\cite{goldwassermicali}. 

In~\cite{bellare2012cryptographic}, it is shown that all the security definitions in the cryptographic context (including a generalization in terms of the mutual information) imply each other, up to some factor. Security for arbitrary messages also implies security for uniform random messages. The converse is not true in general, but can be shown to hold for some classes of channels~\cite{bellare2012cryptographic}.

Explicit wiretap schemes are known using specific error-correcting codes such as low-density parity check codes~\cite{Thangarajldpc} or polar codes~\cite{HofShamaiPolar,MahdavifarVardyPolar}. For the closely related task 
of secure key agreement from a shared resource~\cite{MaurerCommonInformation,AhlswedeCsiszar,maurerwolf} give explicit schemes based on \emph{privacy amplification}~\cite{generalizedpa} with an extractor to reach security.
The concept of extractors can also be applied to secure message transmission over the wiretap channel. In this case, public randomness is used as a seed for the extractor and the extractor needs to be invertible~\cite{CheraghchiDidierShokrollahi09,bellare2012cryptographic,hayashimatsumoto}. The schemes are similar in spirit to previously proposed ones using \emph{syndrome coding} or \emph{coset coding}~\cite{wyner,cohenzemor,cohenzemor06}. 

The schemes based on invertible extractors are \emph{constructive}, i.e., give explicit functions to encode/encrypt and decode/decrypt, and they can be combined with essentially any error-correcting code. They reach strong security for a finite number of channel uses as well as secrecy capacity asymptotically for certain types of channels~\cite{bellare2012cryptographic,hayashimatsumoto}.  In~\cite{tessaroarxiv,semantic} it has been shown that the seed for the extractor can be incorporated into the scheme to obtain an unseeded scheme, while still reaching secrecy capacity. The used functions are computationally efficient and the scheme reaches distinguishing security for any message distribution. 

All of the above results consider the case when the channel uses in the wiretap scheme are identical and memoryless. In real communication channels, however, errors often occur in bursts or may be dependent on external conditions such as the temperature or the weather. Since the channel ultimately has to be implemented physically and depends on physical properties, a real channel will not be perfectly identical or memoryless. In the worst case, the parameters governing the channel properties may even be chosen by the adversary. For the security proof to apply to real schemes, it is therefore necessary to remove these assumptions. 

A notable exception to the assumption of a discrete memoryless channel is the \emph{wiretap channel II}~\cite{wiretap2} (see also~\cite{Goldfeldcuffpermuterwt2}) where the attacker is allowed to obtain an adversarially selected fraction of bits. It dates back to Ozarow and Wyner who evaluate its secrecy capacity  and also give an explicit scheme. In~\cite{varying,WangSafaviAVWTC,blochLaneman}, a generalization of this setup to the arbitrarily varying channel and general wiretap channel is analyzed and the secrecy capacity is evaluated. For the case when channels are selected from a set of types and each type occurs with a fixed frequency,~\cite{varying} gives the secrecy capacity in a single-letter formula; however, no explicit constructive scheme is provided.

\subsection{Our Contribution}

In this paper, we analyze the schemes proposed in \cite{bellare2012cryptographic,tessaroarxiv,semantic} using slightly different techniques to show security. The new analysis has two advantages. First, it yields tighter nonasymptotic bounds for memoryless channels. Second, it applies to certain channel models that are not memoryless, for example, channels where the adversary can choose the behavior of $ChA$ subject to certain constraints; we shall elaborate on below. We therefore make progress towards showing the security of realistic channel models. 

The paper concerns \emph{wiretap} channels, where there is an intended receiver and an eavesdropper (adversary). The channel to the receiver is denoted by $ChR\colon \mathcal{X}\rightarrow \mathcal{Y}$, and the channel to the adversary by $ChA\colon \mathcal{X}\rightarrow \mathcal{Z}$. We shall mainly focus on channels with `good' symmetry properties. In particular, we focus on
adversarial channels $ChA$ where all inputs lead to the same output probability distribution upon relabelling of the values. Additionally, we restrict to 
channels 
for which a uniform input distribution achieves both the Shannon capacity on the receiver's channel and the secrecy capacity of the wiretap channel. That is,
\begin{align*}
\max_{P_X} I(X;Z) = I(X;Z)\Big|_{X\sim\textnormal{uniform}}
\end{align*}
and 
\begin{align*}
    \max_{P_{VX}} I(V;Y)-I(V;Z) = I(X;Y)-I(X;Z)\Big|_{X\sim\textnormal{uniform}}\ ,
\end{align*}
where by $\big|_{X\sim\textnormal{uniform}}$ we mean that the quantities are computed for a uniformly distributed $X$. 
Since the scheme can be combined with a variety of error-correcting codes, we do not focus on correctness. However, the above properties will allow us to use a \emph{linear code} to achieve Shannon capacity to the receiver, and, furthermore, to combine this code with an extractor to yield a secret encoding scheme for the wiretap channel that achieves secrecy capacity. 

We will revisit the explicit efficient scheme in~\cite{tessaroarxiv,semantic}, where the sender encrypts (encodes) a message using the `inverse' of an extractor and then applies an error-correcting code. The receiver first decodes the received value and then applies the extractor to obtain the message. We give a new simple security proof of this scheme and show that it remains secure even for channels where the individual channel uses can differ or have a memory between the individual runs. The adversary is allowed to choose the exact channel from a set, in particular, the adversary can always choose the \emph{order} of the channels. 

The following are our main results:
\begin{itemize}
\item We prove security for random-messages for channels where every input leads to the same output distribution (upon relabelling) and where the output distribution follows an asymptotic equipartition property. This is the case for many distributions which are not identical and independent/memoryless. The order of the channel can be chosen by the adversary (Lemma~\ref{lemma:general_security_bound}, p.~\pageref{lemma:general_security_bound}). The scheme can be combined with \emph{any} error-correcting code to ensure correctness. Our proof is easy-to-understand and reaches positive key rates on previously unattained parameter regions, as well as better rates for finite-length schemes than previous bounds (Figure~\ref{fig:comparison-bounds}, p.~\pageref{fig:comparison-bounds}). We also show that we can reach secrecy capacity when the receiver channel and the adversary's channel both reach capacity for uniform inputs (Lemma~\ref{lemma:capacity_reaching}, p.~\pageref{lemma:capacity_reaching}).
\item We give a general reduction from security for uniform random messages to security for arbitrary message distributions. Our reduction applies to schemes with a linear inverter of an extractor combined with a linear error-correcting code. The individual channel runs are on any alphabet of the form $\mathbb{Z}/p$; they are memoryless and symmetric but do not need to be identical (Theorem~\ref{th:distinguishing_security}, p.~\pageref{th:distinguishing_security}).
\item We use this technique to prove security for arbitrary message distributions of the arbitrarily varying wiretap channel with type-constrained states. The adversary is allowed to choose the state sequence. The scheme reaches secrecy capacity for strongly symmetric individual channels (under the condition that the linear error-correcting code reaches Shannon capacity for the receiver channel) (Theorem~\ref{th:securityforavwtc}, p.~\pageref{th:securityforavwtc}).
\end{itemize}

\paragraph{Outline: }
In Section~\ref{sec:preliminaries}, we present some necessary definitions and well-known theorems, explain the security model, and review the schemes. Our main contribution is in Section~\ref{sec:result}, containing the security proof. We provide the theoretical analysis which we tighten step-by-step and, in parallel, show its usefulness by applying it to a variety of examples of specific channels such as the binary symmetric channel, the wiretap channel II and arbitrarily varying wiretap channel. Section~\ref{sec:achievable_rate} analyzes the rate reached this way and shows that it reaches capacity in many cases.

While we consider random-message security in Section~\ref{sec:result}, we remove this condition in Section~\ref{sec:distinguishing_security} and characterize under which conditions this implies security for arbitrary message distributions, generalizing a statement from~\cite{bellare2012cryptographic} to channels which are not binary or identical.
Finally, Section~\ref{sec:conclusion} gives a conclusion and outlook.

\section{Background and Previous Results}\label{sec:preliminaries}

\subsection{Notation}
We denote random variables by capital letters, such as $X$, the sets of their possible values by calligraphic letters, like $\mathcal{X}$, the cardinality of such sets by $\abs{\mathcal{X}}$, and their realizations by lower-case letters like $x$. The probability that the random variable $X$ takes value $x$ is $P_X(x)$. Sometimes we drop the index when the random variable is clear from the context. All the random variables in this paper take values in discrete sets of finite cardinality.

To specifically emphasise that a random variable consists of $n$ symbols from a set $\mathcal{Y}$, we denote it by $\nfoldY{=}(Y_0, \ldots, Y_{n-1})$ and its value by $\nfoldy{=}(y_0, \ldots, y_{n-1})$.

The \emph{joint probability distribution} of two (or more) random variables $X$ and $Y$ is denoted by $P_{XY}(x,y)$. 
The \emph{conditional probability} of $Y=y$  given $X{=}x$ with $P_X(x){>}0$ is  
$P_{Y|X=x}(y){=}\frac{P_{XY}(x,y)}{P_X(x)}$ and the \emph{conditional probability distribution} $P_{Y|X}$ is  
$P_{Y|X}(y,x){=} P_{Y|X=x}(y)$. 
A conditional probability distribution is similar to a \emph{stochastic kernel} taking the random variable $X$ as input and giving a (probabilistic) output $Y$, depending on the input $X=x$. 

Two random variables $X$ and $Y$ are called \emph{independent} if  
$P_{XY}(x,y){=}P_{X}(x)\cdot P_{Y}(y)$ for all $x,y$.

\begin{mdframed}
Let us denote by $\vv{p}(X|Z)$ the vector which contains the probabilities $P_{X|Z}(x,z)$ ordered by values with the first one being the largest. I.e., the vector $\vv{p}(X|Z)$ contains $n=\abs{\mathcal{X}}{\cdot} \abs{\mathcal{Z}}$ elements ${p_i}(X|Z)$, such that ${p_0}(X|Z){:=}\max_{x,z}P_{X|Z}(x,z)$ and ${p_0}(X|Z){\geq} {p_1}(X|Z){\geq }\ldots{\geq} {p_{n-1}}(X|Z)$, where the ordering is over the elements iterating through all random variables. In contrast, the vector $\vv{p}(X|Z{=}z)$ only contains the $\abs{\mathcal{X}}$ elements for this fixed value of $Z{=}z$ and the ordering is only over $x\in \mathcal{X}$.

With this notation, distributions containing the same probabilities, but associated with different symbols, correspond to the same vector.
\end{mdframed}

The \emph{expected value} of  a random variable $X$ is 
$\mathbf{E}_X (X) {=}\sum_{x\in \mathcal{X}}P_X(x)\cdot x$. The \emph{entropy} of a random variable $X$ is $\mathrm{H}(X){=}-\sum_{x\in\mathcal{X}}P_X(x)\log_2 P_X(x)$. The \emph{binary entropy function} of $p$ is the entropy of a Bernoulli distribution with parameter $p$: $h(p){=}-p\log_2p-(1-p)\log_2(1-p)$.

A special probability distribution is the \emph{uniform} distribution, i.e., the one where all possible outcomes are equally likely, defined as $P_U(u) = \frac{1}{|\mathcal{U}|}$. 
We will often use the letter $U$ (for `uniform') to denote a random variable which is uniformly distributed. 
A random variable $F$ which is drawn uniformly at random from a set $\mathcal{F}$ will be denoted by $F\in_R \mathcal{F}$.

\subsection{Properties of Distributions}
A useful measure of how different two distributions $P$ and $Q$ on the same set $\mathcal{X}$ are is the \emph{relative entropy} $D(P||Q){=}\sum_{x\in\mathcal{X}}P_X(x)\log_2\frac{P_X(x)}{Q_X(x)}$. A proper distance for distributions is the \emph{variational distance}, the minimal probability that a random variable drawn from one or the other distribution takes a different value. 
\begin{definition}
Let $P$ and $Q$ be distributions 
over $\mathcal{X}$. The \emph{variational distance} (also called \emph{statistical distance}) between $P$ and $Q$ is
\begin{align}
\nonumber d(P,Q)&=\frac{1}{2} \sum_{x\in \mathcal{X}}\Big|P(x)-Q(x) \Big|\ .
\end{align}
Two distributions $P$ and $Q$ with variational distance at most $\varepsilon$ are called \emph{$\varepsilon$-close} and the set of all $\varepsilon$-close distributions to a certain distribution $P$ is denoted by $\mathcal{P}_P^{\varepsilon}$, i.e., 
\begin{align*}
  \mathcal{P}_P^{\varepsilon}&=\{Q'\;|\; d(P,Q')\leq \varepsilon\}  .
\end{align*}
The \emph{conditional distance} of $P$ and $Q$ given a random variable $W$ is the expectation of the distance over $W$ 
\begin{align}
\nonumber d(P,Q|W)& = \frac{1}{2} \sum_{w\in \mathcal{W}}P_W(w)\left(\sum_{x\in \mathcal{X}}\Big|P_{X|W=w}(x)-Q_{X|W=w}(x) \Big|\right)\ .
\end{align}
\end{definition}
Of particular importance to us is the distance of a distribution $P_V$ from the uniform one. We denote this distance by $d_U(V)$. 
\begin{definition}\label{def:distancefromuniform}
The \emph{distance from uniform} of a random variable $V$ over $\mathcal{V}$ with distribution $P_V$ is the variational distance between $P_V$ and the uniform distribution over  $\mathcal{V}$, i.e.,
\begin{align}
\nonumber d_U(V)&=\frac{1}{2} \sum_{v\in \mathcal{V}}\left|P_V(v)-\frac{1}{|\mathcal{V}|}\right|\ .
\end{align}
\end{definition}

The min-entropy of $V$ given $Z$ is related to the maximal probability that someone receiving the value of $Z$ can correctly guess the value of $V$. The $\varepsilon$-smooth version of it is defined as the \emph{largest} min-entropy of any distribution which is $\varepsilon$-close to the original one.
\begin{definition}[Guessing probability of $V$ given $Z$]
The \emph{guessing probability} of $V$ given $Z$ of a joint distribution $P_{VZ}(v,z)$  is
\begin{align}
\nonumber P_\mathrm{guess}(V|Z) &= \mathbf{E}_Z \max_v  P_{V|Z}(v,Z).
\end{align}
The \emph{$\varepsilon$-guessing probability} of $V$ given $Z$ is the minimal guessing probability of $V$ given $Z$ of all joint distributions $\tilde{P}_{VZ}$ which are $\varepsilon$-close to $P_{VZ}(v,z)$
\begin{align}
\nonumber P^{\varepsilon}_{\mathrm{guess}}(V|Z)&= \min_{\tilde{P}_{VZ}\in \mathcal{P}_P^{\varepsilon}}\left( \mathbf{E}_Z \max_v  \tilde{P}_{V|Z}(v,Z)\right)\ .
\end{align}
\begin{align*}
\end{align*}
\end{definition}

\begin{definition}[min-entropy of $V$ given $Z$]
The \emph{min-entropy} of $V$ given $Z$ of a joint distribution $P_{VZ}(v,z)$  is
\begin{align}
\nonumber \mathrm{H}_{\mathrm{min}}(V|Z)&= -\log_2 P_\mathrm{guess}(V|Z)\ .
\end{align}
The \emph{$\varepsilon$-smooth min-entropy} is 
\begin{align}
\nonumber \mathrm{H}^{\varepsilon}_{\mathrm{min}}(V|Z)&= -\log_2 P^{\varepsilon}_{\mathrm{guess}}(V|Z)\ .
\end{align}
\begin{align*}
\end{align*}
\end{definition}

We shall use the asymptotic equipartition property (AEP) for sequences of i.i.d. random variables. 
\begin{theorem}[Asymptotic equipartition property (see, e.g.~\cite{coverthomas})] Let $Z_1$,\ldots, $Z_n$ be i.i.d. according to $P_Z$. Then, for any $\varepsilon>0$, 
\begin{align}\label{eq:AEP}
\Pr[ \prod_i P(Z_i)> 2^{-n(H(Z)-\varepsilon)}] &\leq \frac{\mathrm{Var}[-\log_2 P_Z(Z)]}{n\varepsilon^2}\ ,  
\end{align}
where $Z$ is distributed according to $P_Z$.
\end{theorem}

The AEP can be extended to some non-i.i.d.\ cases, for example, where $Z_1,\ldots,Z_n$ is a stationary ergodic process; see, e.g., \cite[Ch.~15]{coverthomas}. In this case, the entropy $\mathrm{H}(Z)$ would be replaced by the \emph{entropy rate} of the random process, i.e., $\lim_{n\rightarrow\infty}\frac{1}{n}\mathrm{H}(Z_1,\dots,Z_n)$. Another useful scenario is where $Z_1,\ldots,Z_n$ are independently drawn according to a conditional distribution $P_{Z|S}$ conditional on a state sequence $s_1,\ldots,s_n$, and where the sequence $s_1,\ldots,s_n$ is a of a given \emph{type}, namely, the number of occurrences of every possible state is fixed and known, but their order may be arbitrary. In such cases, one can write bounds similar to \eqref{eq:AEP}. 

\subsection{Extractors}
\begin{definition}[Strong extractor]
A function \emph{$\text{Ext}\colon\mathcal{V}\times \mathcal{S}\rightarrow \mathcal{M}$} is a \emph{$(k,\varepsilon)$-strong extractor} if for all random variables $V\in \mathcal{V}$ and $Z\in\mathcal{Z}$ with $\mathrm{H}_{\mathrm{min}}(V|Z)\geq k$ and an independent seed $S\in_R \mathcal{S}$ chosen uniformly at random it holds that
\begin{align}
\nonumber d_U(\emph{\text{Ext}}(V,S),S|Z)&\leq \varepsilon\ .
\end{align}
\end{definition}

\begin{definition}[Two-universal function family]
A set of functions $f\colon\mathcal{V}\rightarrow \mathcal{M}$ for $f\in \mathcal{F}$ is called \emph{two-universal} if  for any $v_0\neq v_1 \in \mathcal{V}$
\begin{align*}
\sum_{f\in \mathcal{F}}\frac{1}{\abs{\mathcal{F}}}\mathbbm{1}[f(v_0) = f(v_1)] &\leq \frac{1}{\abs{\mathcal{M}}}\ ,
\end{align*}
where $\mathbbm{1}[\textnormal{statement}]$ equals $1$ when the statement is true and equals zero otherwise.
\end{definition}

Two-universal functions can be seen as extractors where the seed selects a particular function from the family, i.e., $f_s(v){:=}\text{Ext}(v,s)$. Two-universal functions are good strong extractors, as stated by the left-over hash lemma from~\cite{ILL,HILL}. 
We state it here in the version from~\cite{biometrics}.
\begin{theorem}[Left-over hash lemma~\cite{ILL,HILL,biometrics}]\label{thm:lohl}
Let $f_{\mathcal{S}}\colon \mathcal{V}\rightarrow \mathcal{M}$ be a two-universal hash function with $\mathcal{S}$ indicating the set of functions. Let $V$ be a random variable on $\mathcal{V}$, potentially correlated with a second random variable $Z\in \mathcal{Z}$. 
Then, for $S\in_R \mathcal{S}$,
\begin{align*}
d_U(f_S(V),S|Z)&\leq \frac{1}{2}\sqrt{\abs{\mathcal{M}}2^{-\mathrm{H}_{\mathrm{min}}(V|Z)}}\ .
\end{align*}
\end{theorem}
Using the definition of the min-entropy, this bound is directly related to the guessing probability. We will mostly use the left-over hash lemma in terms of the $\varepsilon$-smooth min-entropy: it follows from Theorem~\ref{thm:lohl} that
\begin{align}
\label{eq:epsleftoverhash}
d_U(f_S(V),S|Z)&\leq \frac{1}{2}\sqrt{\abs{\mathcal{M}}2^{-\mathrm{H}^{\varepsilon}_{\mathrm{min}}(V|Z)}}+\varepsilon\ .
\end{align}

We will be interested in hash functions which can be \emph{inverted}, i.e., for which we can efficiently find preimages for a given seed, using additional randomness as input. 
\begin{definition}[Inverter]\label{def:inverter}
Let \emph{$\text{Ext}:\mathcal{V}\times \mathcal{S}\rightarrow \mathcal{M}$} be a strong extractor. Then \emph{$\text{Inv}:\mathcal{M} \times \mathcal{S} \times \mathcal{R} \rightarrow \mathcal{V}$} is an \emph{inverter} of $\mathrm{Ext}$ if for all $m\in \mathcal{M}$, $s\in \mathcal{S}$ and a uniformly chosen $R\in_R \mathcal{R}$, the distribution is uniform over all preimages of $m$, i.e., 
\begin{align*}
    \emph{\text{Inv}}(m,s,R)\in_R \lbrace v\in \mathcal{V}\ |\, \emph{\text{Ext}}(v,s)=m\rbrace
\end{align*}
\end{definition}

In practice, two-universal function families which are often used in the context of the left-over hash lemma include multiplication 
of a bit-string $v$ with a randomly chosen matrix over $\text{GF}(2)$~\cite{carterwegman} or multiplication of the bit-string $v$ with a randomly chosen Toeplitz matrix over $\text{GF}(2)$~\cite{Toeplitzhashing}. In the context of the present use case, we are only interested in two-universal functions for which an inverter exists. This is the case for multiplication with a random element or multiplication with a randomly chosen modified Toeplitz matrix:
\begin{enumerate}
\item Finite field extractor~\cite{carterwegman}: An $l$-bit string can be thought of as an element of the extension field $\text{GF}(2^l)$. For a uniform seed $S\in_R\text{GF}(2^l) \setminus \{0\}$, this extractor outputs the first $\lambda$ bits (denoted by $\big|_\lambda$) of the input times the seed using the finite field multiplication (denoted by $\star$), i.e.,
\begin{align*}
\text{Ext}\colon & \{0,1\}^l \times \{0,1\}^l \rightarrow \{0,1\}^\lambda \\
&(v, s) \mapsto v \star s\Big|_\lambda \ .
\end{align*}
An inverter of this extractor is
\begin{align*}
\text{Inv}\colon & \{0,1\}^\lambda \times \{0,1\}^l \times \{0,1\}^{l-\lambda} \rightarrow \{0,1\}^l \\
&(m, s, r) \mapsto s^{-1} \star (m\,||\,r)\,,
\end{align*}
where $||$ means concatenation, and the inverse of the seed is with respect to the finite field multiplication.   
This extractor and its inverter can be implemented efficiently with complexity $\mathcal{O}(n\cdot\log n\cdot\log\log n)$ using the Sch\"onhage-Strassen algorithm~\cite{Schnhage1971} (see, e.g., \cite{GMP-SS-alg}).
\item Modified Toeplitz hashing~\cite{hayashi-toeplitz}: A seed $s\in\text{GF}(2^{l-1})$ can be used to construct an $\lambda\times(l-\lambda)$ Toeplitz matrix $T(s)$. The extractor is defined as the matrix-vector multiplication of the input with the Toeplitz matrix concatenated with an identity matrix $\id_\lambda$, i.e.,
\begin{align*}
\text{Ext}\colon & \{0,1\}^l \times \{0,1\}^{l-1} \rightarrow \{0,1\}^\lambda \\
&(v,s) \mapsto \begin{bmatrix}T(s) & \id_\lambda\end{bmatrix}v
.\end{align*}
A corresponding inverter can be constructed with additional randomness $r\in\{0,1\}^{l-\lambda}$ as
\begin{align*}
\text{Inv}\colon & \{0,1\}^\lambda \times \{0,1\}^{l-1} \times \{0,1\}^{l-\lambda} \rightarrow \{0,1\}^l \\
&(m, s, r) \mapsto \begin{bmatrix}\id_{l-\lambda} & \mathbf{0} \\ -T(s) & \id_\lambda \end{bmatrix}\begin{bmatrix}r \\ m\end{bmatrix} = \begin{bmatrix}r \\ -T(s)r + m\end{bmatrix}\ .
\end{align*}
We can check that this is indeed the inverter of the above extractor:
\begin{align*}
\text{Ext}\big(\text{Inv}(m,s,r)\big)&=\begin{bmatrix}T(s) & \id_\lambda \end{bmatrix}\begin{bmatrix}r \\ -T(s)r + m\end{bmatrix} = T(s)r-T(s)r+m=m.
\end{align*}
There exist efficient matrix-vector multiplication algorithms to implement this extractor with complexity $\mathcal{O}(n\cdot\log n)$ (see, e.g.\ Sec.~4.8 from \cite{golub2013matrix}).
\end{enumerate}

\subsection{Codes and Channels}
We recall some basic concepts of error-correcting codes. See, e.g.~\cite{coverthomas,loeliger} for more details.
\begin{definition}[Error-correcting code]
Let $\mathcal{A}$ be a nonempty set and let $n$ be a positive integer. An \emph{error-correcting code} over $\mathcal{A}$ is a nonempty subset $C\subseteq \mathcal{A}^n$. The integer $n$ is called \emph{length} of the code. The elements of the code $c \in C$ are called \emph{codewords}. The set $\mathcal{A}$ is called the \emph{alphabet} of the code. 
\end{definition}

\begin{definition}[Linear error-correcting code]
A \emph{linear error-correcting code} over a finite field $F$ is a subspace of the vector space $F^n$. 
\end{definition}

We will use the fact that the number of elements $\abs{F}$ of a finite field $F$ is $p^k$ for some prime number $p$. We will also consider fields of the form $\mathbb{Z}/p$ for a prime $p$, which contain $p$ elements. 
Note that the vector space $F^n$ contains $\abs{F}^n$ elements and the number of elements $\abs{C}$ in a subspace $C$ divides the number of elements of $F^n$. Furthermore, any projection onto the first $k$ dimensions (symbols) of a subspace $C$ results in a subspace of $F^k$, the number of elements therefore divides $\abs{F}^k$.

\begin{definition}[Encoding]
The \emph{encoder} of an error-correcting code $C$ is a bijective map from $\kprodvset\rightarrow C \subseteq \mathcal{A}^n$. 
\end{definition}

\begin{definition}[Decoding] A \emph{decoder} of an error-correcting code $C\subseteq \mathcal{A}^n$ is a function $\mathcal{A}^n\rightarrow \kprodvset$.
\end{definition}

\begin{definition}[Channel] 
A \emph{channel} $ChA:\mathcal{X}\rightarrow \mathcal{Z}$ is described by a conditional probability distribution $P_{Z|X}(z,x)$. This conditional probability distribution is sometimes denoted by a matrix $W$ (or $W(z|x)$) called the \emph{transition matrix} of the channel. 
\end{definition}
To emphasize that the distribution of a random variable, or alternatively an element drawn from this distribution, is obtained from sending an input $x$ through the channel $ChA$, we sometimes denote this by $ChA(x)$.

\begin{definition}[Symmetric channel]
A \emph{channel} $ChA:\mathcal{X}\rightarrow \mathcal{Z}$ is called \emph{strongly symmetric} when all rows and columns of the transition matrix $W$ are permutations of each other. It is called \emph{symmetric} when there exists a partition of the outputs $\mathcal{Z}{=} \bigcup_v \mathcal{Z}_v$ such that each submatrix of $W$ induced by an element of the partition is strongly symmetric, i.e., for every $x\neq x'$ and $z\neq z'$ with $z,z'\in \mathcal{Z}_{\bar{v}}$, there exist permutations $\pi^{x\mapsto x'}:z\mapsto z'$ and 
$\pi^{z\mapsto z'}:x\mapsto x'$ such that 
\begin{align*}
 W(z|x) &= W(\pi^{x\mapsto x'}(z)|x') = W(z'|\pi^{z\mapsto z'}(x)) =  W(z'|x').
\end{align*}
\end{definition}

For a channel that acts on a sequence of $n$ input symbols and gives $n$ output symbols, we sometimes write $ChA^{(n)}$. Of special interest are channels which consist of $n$ channels, each only acting locally on the $i$'th input and output symbol. When additionally all the $n$ channels are the same, this is usually called a \emph{discrete memoryless channel}. We add the specification `identical' or `not necessarily identical' to distinguish the two cases. 

\begin{definition}
An $n$-fold \emph{not necessarily identical memoryless channel} $ChA^{(n)}$ is a channel of the form
\begin{align*}
    P_{\nfoldZ|\nfoldX}(\nfoldz,\nfoldx)&= \prod_{i=0}^{n-1} P_{Z_i|X_i}(z_i,x_i)\ .
\end{align*}
When additionally $P_{Z_i|X_i}(z_i,x_i)=P_{Z_j|X_j}(z_j,x_j)$ for all $i,j$, the channel is called \emph{identical memoryless channel} and we may write $ChA^{(n)}=\bigotimes_i ChA_i$. 
\end{definition}

We will often first apply a (randomized) function $f:\mathcal{M}\rightarrow \nprodxset$, such as an inverter or an error-correcting code, to a message and then apply a channel $ChA^{(n)}:\nprodxset\rightarrow \nprodzset$ to the result of the function. This defines a new channel, which we denote by $Ch{=}ChA^{(n)}\circ f:\mathcal{M}\rightarrow \nprodzset$.

\subsection{Modelling Secure Message Transmission in the Wiretap Scenario}

We will compare our \emph{real} cryptographic system to an \emph{ideal} system which is secure by construction~\cite{pw,bpw,canetti}  using the framework of \emph{random systems}~\cite{Maurer02}. This naturally leads to  \emph{distinguishing security} as metric, however, by~\cite{bellare2012cryptographic}, this is equivalent to other commonly used metrics. 

A \emph{system} is an abstract device taking inputs and
giving outputs at one or more \emph{interfaces} and is characterized by the 
probability distributions of the outputs given the inputs. 
The closeness of two systems $\mathcal{S}_0$ and $\mathcal{S}_1$ is measured by introducing an additional system called \emph{distinguisher}. The distinguisher $\mathcal{D}$ interacts with another system guessing which system it is connected to. 

The \emph{distinguishing advantage between systems $\mathcal{S}_0$ and $\mathcal{S}_1$} is defined in terms 
of the probability of correctly recognizing the system when connected to one of the two at random.  
\begin{definition}
	The \emph{distinguishing advantage between two systems $\mathcal{S}_0$ and $\mathcal{S}_1$ }is 
	\begin{eqnarray}
		\nonumber \delta(\mathcal{S}_0, \mathcal{S}_1)&=& \max_{\mathcal{D}}[P(B=1|\mathcal{S}=\mathcal{S}_1)-P(B=1|\mathcal{S}=\mathcal{S}_0)]\ ,
	\end{eqnarray}
	where the maximum ranges over all distinguishers $\mathcal{D}$ connected to a system $\mathcal{S}$ and where $B$ denotes 
	the output of the distinguisher. 
	Two systems $\mathcal{S}_0$ and $\mathcal{S}_1$ are called \emph{$\epsilon$-indistinguishable} if $\delta(\mathcal{S}_0, \mathcal{S}_1){\leq} \epsilon$.
\end{definition}

The probability of any event $\mathcal{E}$, defined in a scenario involving the ideal system $\mathcal{S}_0$ cannot 
differ by more than this quantity from the probability of a corresponding event in a scenario where $\mathcal{S}_0$ 
has been replaced by the real system $\mathcal{S}_1$ and, therefore, the resulting security is \emph{composable}~\cite{pw,bpw,canetti}. 

The distinguishing advantage is a \emph{pseudo-metric}, in particular, it fulfils the triangle inequality 
\begin{eqnarray}
\nonumber	\label{eq:triangle} \delta(\mathcal{S}_0,\mathcal{S}_1)+\delta(\mathcal{S}_1,\mathcal{S}_2)&\geq & \delta(\mathcal{S}_0,\mathcal{S}_2)\ .
\end{eqnarray}

The ideal system for secure message transmission is one where the sender inputs a message $m$, the receiver receives the same message $m$ and the eavesdropper receives nothing, or, to be precise, receives outputs which are independent from the message transmission system.\footnote{Since the outputs are uncorrelated to the actual system the adversary can simulate these output distributions themselves and does therefore not gain anything compared to attacking a system which outputs nothing.} This setup is depicted in Figure~\ref{fig:wiretap}. In contrast to many other works on wiretap security, we allow the adversary to choose certain channel properties. This is modelled by the input $W$ and reflects the choice of a strategy by the adversary. The adversary then receives the output of the adversarial channel denoted by the random variable $Z$, which will, in general, depend on the input. To indicate this dependency, we will sometimes use superscripts in probability distributions or guessing probabilities, e.g. $P^w(Z)$ or $P^{\varepsilon\ w}_{\mathrm{guess}}(\kfoldV|\nfoldZ)$. The adversary additionally obtains the seed value $S=s$.   

\begin{figure}[ht] \label{fig:ideal}
	\centering
	\begin{tikzpicture}[scale=1.0]   
            \path[clip] (-3.8,-1.0) rectangle (3.8,3.5);
		\node[align=left, anchor=west] at (-3,3.3) {The ideal system:};
		\node[align=right, anchor=south west, rotate=90] at (-3.1,0) {\small{SENDER}};
		\node[align=left, anchor=south east, rotate =-90] at (3.1,0) {\small{RECEIVER}};
		\node[align=left, anchor=west] at (-3,-0.3) {\small{EAVESDROPPER}};
		\draw[line width=1mm] (-3,0) rectangle (3,3);
		\draw[<-, line width=0.2mm, color=black] (-2.2,2.6) -- (-3.2,2.6) node[anchor = east]{$m$};	
		\draw[->, line width=0.2mm, color=black] (2.2,2.6) -- (3.2,2.6) node[anchor = west] {$m$};	
		\draw[|->, line width=0.1mm, color=black] (0.9,-0.2) -- (0.9,-0.5) node[anchor = north] {$Z$};
    \draw[|->, line width=0.1mm, color=gray] (0.5,-0.2) -- (0.5,-0.5) node[anchor = north] {$S$};
        \draw[<-, line width=0.1mm, color=gray] (0.1,0.5) -- (0.1,-0.5) node[anchor = north] {$W$};
	\end{tikzpicture}
	\begin{tikzpicture}[scale=1.0]
             \path[clip] (-3.95,-1.0) rectangle (3.95,3.5);
		\node[align=left, anchor=west] at (-3,3.3) {The real system:};
		\node[align=right, anchor=south west, rotate=90] at (-3.1,0) {\small{SENDER}};
		\node[align=left, anchor=south east, rotate =-90] at (3.1,0) {\small{RECEIVER}};
		\node[align=left, anchor=west] at (-3,-0.3) {\small{EAVESDROPPER}};
		\draw[line width=1mm] (-3,0) rectangle (3,3);
  		\draw[<-, line width=0.2mm, color=black] (-2.2,2.6) -- (-3.2,2.6) node[anchor = east]{$m_A$};	
		\draw[->, line width=0.2mm, color=black] (2.2,2.6) -- (3.2,2.6) node[anchor = west] {$m_B$};
		\draw[line width=0.5mm] (-1,0.3) rectangle (1,0.8)  node[pos=.5, text width=3cm, align=center] {\small{wiretap ch.}};
		\node[align=center, anchor=south] at (-2.0,2.1) {\small{seed ${\in_R} \mathcal{S}$}};
		\draw[->, line width=0.1mm, color=black] (-2.2,2.3) -- (-2.2,2.0) node[anchor = north, inner sep=0em] {\small{$\mathrm{INV}$}};
		\draw[->, line width=0.1mm, color=black] (-2.2,1.7) -- (-2.2,1.4) node[anchor = north, inner sep=0em] {\small{$\mathrm{ENC}$}};
		\draw[->, line width=0.1mm, color=black] (-2.2,1.1) -- (-2.2,0.8) node[anchor = north, inner sep=0em] {$x||s$};		
		\draw[<-, line width=0.1mm, color=black] (-1,0.55) -- (-1.7,0.55);
		\draw[->, line width=0.1mm, color=white] (2.2,2.8) -- (2.2,2.4) node[anchor = north, inner sep=0em] {\textcolor{black}{\small{$\mathrm{EXT}$}}};
\draw[<-, line width=0.1mm, color=black] (2.2,2) -- (2.2,1.6) node[anchor = north, inner sep=0em] {\small{$\mathrm{DEC}$}};
\draw[<-, line width=0.1mm, color=black] (2.2,1.2) -- (2.2,0.8) node[anchor = north, inner sep=0em] {$y||s'$};
		\draw[->, line width=0.1mm, color=black] (1,0.55) -- (1.7,0.55);
		\draw[->, line width=0.1mm, color=black] (0.9,0.3) -- (0.9,-0.5) node[anchor = north] {$Z$};
    \draw[->, line width=0.1mm, color=gray] (0.5,0.3) -- (0.5,-0.5) node[anchor = north] {$S$};
        \draw[<-, line width=0.1mm, color=gray] (0.1,0.3) -- (0.1,-0.5) node[anchor = north] {$W$};
	\end{tikzpicture}
	\caption{\label{fig:wiretap} The \emph{real} (right) and \emph{ideal} (left) message transmission system. The \emph{ideal} system $\mathcal{S}_{\mathrm{ideal}}$ outputs the sender's message to the receiver and nothing to the eavesdropper. }
\end{figure}
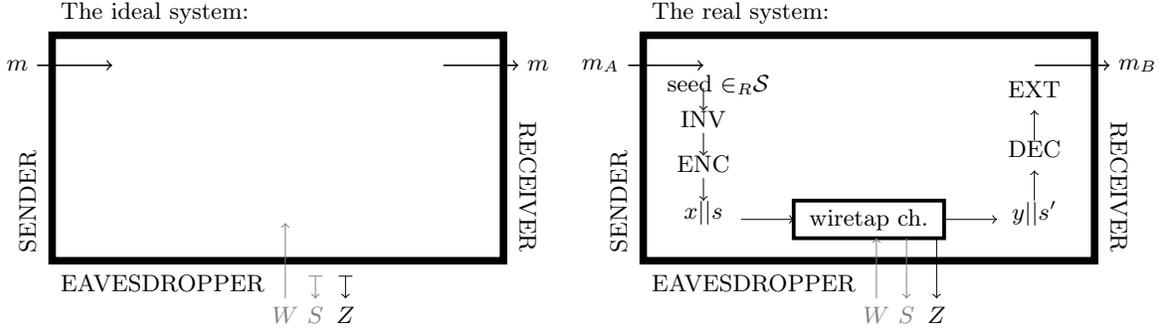

\begin{definition}
    The \emph{ideal message transmission system} with adversarial input takes a message $m\in \mathcal{M}$ as input on the sender's side and outputs the same $m$ on the receiver's side. It takes an input $w\in \mathcal{W}$ on the eavesdropper's side and output's the dummy symbol.
\end{definition}

Due to the triangle inequality on systems, it is possible to introduce an intermediate system and divide the requirements for secure message transmission into two aspects (see Figure~\ref{fig:secandcorr}): 
\begin{itemize}
    \item \emph{correctness} (or \emph{decoding}): the receiver obtains the correct message, and
    \item \emph{secrecy}: the eavesdropper learns nothing about the message.
\end{itemize}
A real system which is $\epsilon_{cor}$-correct and $\epsilon_{sec}$-secret is $\epsilon$-secure, with $\epsilon{=}\epsilon_{cor}{+}\epsilon_{sec}$.

\begin{figure}[ht]
	\centering
	\begin{tikzpicture}[scale=0.6]   
            \path[clip] (-4.5,-1.3) rectangle (22.5,4.2);
		\draw[<->, line width=0.1mm, color=black] (3.2,-0.5) -- (5.8,-0.5) node[pos=0.5,anchor = north]{$\epsilon_{sec}$};   
  		\draw[<->, line width=0.1mm, color=black] (12.2,-0.5) -- (14.8,-0.5) node[pos=0.5,anchor = north]{$\epsilon_{cor}$};
		\node[align=left, anchor=west] at (-4.5,3.5) {Ideal system:};
		\draw[line width=1mm] (-3,0) rectangle (3,3);
		\draw[<-, line width=0.2mm, color=black] (-2,2) -- (-3.5,2) node[anchor = east, inner sep=0em]{$m$};	
		\draw[->, line width=0.2mm, color=black] (2,2) -- (3.5,2) node[anchor = west, inner sep=0em] {$m$};	
		\draw[|->, line width=0.1mm, color=black] (1.7,-0.2) -- (1.7,-0.5) node[anchor = north] {$Z$};
    \draw[|->, line width=0.1mm, color=gray] (0.9,-0.2) -- (0.9,-0.5) node[anchor = north] {$S$};
        \draw[<-, line width=0.1mm, color=gray] (0.1,0.7) -- (0.1,-0.5) node[anchor = north] {$W$};
  \begin{scope}[shift={(9,0)}]
		\node[align=left, anchor=west] at (-4.5,3.5) {Intermediate system:};
		\draw[line width=1mm] (-3,0) rectangle (3,3);
		\draw[<-, line width=0.2mm, color=black] (-2,2) -- (-3.5,2) node[anchor = east, inner sep=0em]{$m$};	
		\draw[->, line width=0.2mm, color=black] (2,2) -- (3.5,2) node[anchor = west, inner sep=0em] {$m$};	
		\draw[->, line width=0.1mm, color=black] (1.7,0.7) -- (1.7,-0.5) node[anchor = north] {$Z$};
    \draw[->, line width=0.1mm, color=gray] (0.9,0.7) -- (0.9,-0.5) node[anchor = north] {$S$};
        \draw[<-, line width=0.1mm, color=gray] (0.1,0.7) -- (0.1,-0.5) node[anchor = north] {$W$};
        \end{scope}
  \begin{scope}[shift={(18,0)}]
		\node[align=left, anchor=west] at (-4.5,3.5) {Real system:};
		\draw[line width=1mm] (-3,0) rectangle (3,3);
		\draw[<-, line width=0.2mm, color=black] (-2,2) -- (-3.5,2) node[anchor = east, inner sep=0em]{$m$};	
		\draw[->, line width=0.2mm, color=black] (2,2) -- (3.5,2) node[anchor = west, inner sep=0em] {$m'$};;	
		\draw[->, line width=0.1mm, color=black] (1.7,0.7) -- (1.7,-0.5) node[anchor = north] {$Z$};
    \draw[->, line width=0.1mm, color=gray] (0.9,0.7) -- (0.9,-0.5) node[anchor = north] {$S$};
        \draw[<-, line width=0.1mm, color=gray] (0.1,0.7) -- (0.1,-0.5) node[anchor = north] {$W$};
        \end{scope}
	\end{tikzpicture}
	\caption{\label{fig:secandcorr} We introduce an intermediate system which replaces the receiver's output by the input.}
\end{figure}
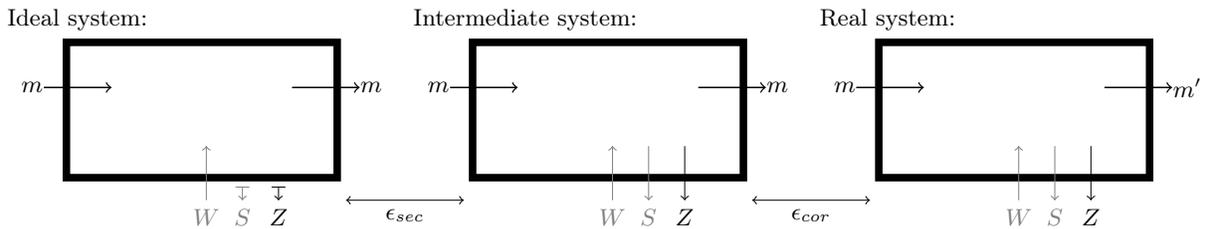

In this paper, we consider secrecy and correctness in two situations: one where the input message $m$ is chosen uniformly at random from the message space $\mathcal{M}$ (indicated by the superscript `rm') and the other where $m$ is chosen arbitrarily from the message space (indicated by the superscript `mt'). Secrecy and correctness are defined accordingly for random messages or taking the worst message (or message distribution). We note that, to model the former, random-message situation, one should consider $m$ as part of the system, instead of as an input to the system; recall Fig.~\ref{fig:ideal}. For brevity, we shall not give a graphic illustration of this system.

Secrecy and correctness can depend on the adversary's strategy. We require security to hold for \emph{any} adversarial strategy. 
\begin{definition}
    A system $\mathcal{S}_{\mathrm{real}}$ is \emph{$\epsilon^{rm}_{cor}$-correct for random messages} when for any adversarial strategy $w\in \mathcal{W}$
    \begin{align*}
        \sum_m \frac{1}{\abs{\mathcal{M}}} \Pr^w[m'\neq m]&\leq \epsilon^{rm}_{cor}\ .
    \end{align*}
    It is \emph{$\epsilon^{rm}_{sec}$-secret for random messages} when 
    \begin{align*}
       \max_{w\in \mathcal{W}} d_U(M|Z,W=w)&\leq \epsilon^{rm}_{sec}\ .
    \end{align*}
\end{definition}

\begin{definition}
    A system $\mathcal{S}_{\mathrm{real}}$ is \emph{$\epsilon^{mt}_{cor}$-correct for arbitrary messages} when for any adversarial strategy $w\in \mathcal{W}$
    \begin{align*}
        \max_m \Pr^w[m'\neq m]&\leq \epsilon^{mt}_{cor}\ .
    \end{align*}
    It is \emph{$\epsilon^{mt}_{sec}$-secret for arbitrary messages} when
    \begin{align*}
     \max_{w\in \mathcal{W}}  \max_{m,\bar{m}\in \mathcal{M}} d\big((Z(m),S),(Z(\bar{m}),S)|W=w\big)&\leq \epsilon^{mt}_{sec}\ ,
    \end{align*}
    where we denoted by $Z(m)$ the output distribution of the adversarial channel upon input message $m$. 
\end{definition}

\subsection{Protocols from \cite{tessaroarxiv} Using an Inverter of an Extractor}

The explicit scheme from~\cite{tessaroarxiv} for which we will show security is depicted in Figure~\ref{fig:wiretap-scheme} and described in Protocols~\ref{prot:prot} and~\ref{prot:unseeded}. 

The protocol uses as building block a \emph{seeded wiretap channel} (see~\cite{tessaroarxiv} and~\cite{hayashimatsumoto}). This variant of a wiretap channel assumes that all parties have access to a public random seed. The sender then applies the inverse of a strong seeded extractor and an error-correcting code to the message before sending it over the channel. The error-correcting code ensures correctness. The strong seeded extractor ensures that the message looks uniform for any adversary with high enough min-entropy about the message. 

The seeded wiretap scheme can be transformed into an unseeded wiretap scheme by sending the seed over the channel~\cite{tessaroarxiv}. As we shall later see, since the seed can be reused several times, this does not affect the asymptotic rate.

\begin{protocol}[Seeded Wiretap scheme]\label{prot:prot}
Precondition: Sender and receiver have access to a seed $s\in_R \mathcal{S}$.
\begin{enumerate}
\item \label{prot:enum:1} The sender chooses a message $\lfoldm\in \lprodmset$ and randomness $r\in_R \mathcal{R}$
\item \label{prot:enum:2} The sender calculates $\nfoldx{=}\mathrm{ECC}(\mathrm{INV}(\lfoldm,s,r))$
\item \label{prot:enum:3} The sender sends $\nfoldx$ over the channel $ChR^{(n)}$.
\item  \label{prot:enum:4} The receiver obtains $\nfoldy$ according to $\nfoldY=ChR^{(n)}(\nfoldx)$.
\item \label{prot:enum:5}The receiver calculates $m'=\mathrm{EXT}(\mathrm{DEC}(\nfoldy),s)$.
\end{enumerate}
\end{protocol}

\begin{figure}[ht]
\centering
\begin{tikzpicture}	[line width=1pt, scale=1]
			\path[clip] (0.9,0.95) rectangle (5.25,6.9);
   		\node[align=right, anchor=south west, inner sep=0em] at (1,6.5) {\textbf{Seeded scheme:}};
			\draw[xstep=0.2cm, ystep=0.2cm,black,thin] (1.0-0.0001,1.4-0.0001) grid (3.0+0.0001,1.6+0.0001);
   		\node[align=left, anchor=south west] at (1.0,0.8) {message $m\in \lprodmset$};
     	\draw[xstep=0.2cm, ystep=0.2cm,black,thin] (2.6-0.0001,2.2-0.0001) grid (3.4+0.0001,2.4+0.0001);
   		\node[align=left, anchor=south west] at (2.05,1.7) {randomness $r{\in_R}\mathcal{R}$};
          \draw[xstep=0.2cm, ystep=0.2cm,black,thin] (2.6-0.0001,3.0-0.0001) grid (3.4+0.0001,3.2+0.0001);
   		\node[align=left, anchor=south west] at (2.05,2.5) {seed $s{\in_R}\mathcal{S}$};
     			\draw[xstep=0.2cm, ystep=0.2cm,black,thin] (1.0-0.0001,4.4-0.0001) grid (4.0+0.0001,4.6+0.0001);
			\draw[->, line width=0.2mm, color=black] (2,1.6) -- (2,4.4) node[pos=0.9, anchor = west] {$\kfoldv{=}\mathrm{INV}(m,s,r)$};
   			\draw[->, line width=0.2mm, color=black] (2.6,2.3) -- (2.0,2.3);
         			\draw[->, line width=0.2mm, color=black] (2.6,3.1) -- (2.0,3.1);
     			\draw[xstep=0.2cm, ystep=0.2cm,black,thin] (1.0-0.0001,6.0-0.0001) grid (4.6+0.0001,6.2+0.0001);
			\draw[->, line width=0.2mm, color=black] (2.2,4.6) -- (2.2,6) node[pos=0.8, anchor = west] {$\nfoldx{=}\mathrm{ECC}(\kfoldv)$};
\end{tikzpicture}
\begin{tikzpicture}	[line width=1pt, scale=1]
			\path[clip] (0.9,0.9) rectangle (10.6,6.9);
   		\node[align=right, anchor=south west, inner sep=0em] at (1,6.5) {\textbf{Unseeded scheme:}};
   		\node[align=right, anchor=south west, inner sep=0em] at (4,3.5) {\Large{\textbf{\dots}}};
			\draw[xstep=0.2cm, ystep=0.2cm,black,thin] (1.0-0.0001,1.4-0.0001) grid (3.0+0.0001,1.6+0.0001);
   		\node[align=left, anchor=south west] at (1.0,0.8) {$m_0\in \lprodmset$};
     	\draw[xstep=0.2cm, ystep=0.2cm,black,thin] (2.6-0.0001,2.2-0.0001) grid (3.4+0.0001,2.4+0.0001);
   		\node[align=left, anchor=south west] at (2.6,1.7) {$r_0{\in_R}\mathcal{R}$};
     			\draw[xstep=0.2cm, ystep=0.2cm,black,thin] (1.0-0.0001,4.4-0.0001) grid (4.0+0.0001,4.6+0.0001);
			\draw[->, line width=0.2mm, color=black] (2,1.6) -- (2,4.4) node[pos=0.9, anchor = west] {$\kfoldv_0$};
   			\draw[->, line width=0.2mm, color=black] (2.6,2.3) -- (2.0,2.3);
         			\draw[->, line width=0.2mm, color=black] (5.8,3.1) -- (2.0,3.1);
     			\draw[xstep=0.2cm, ystep=0.2cm,black,thin] (1.0-0.0001,6.0-0.0001) grid (4.6+0.0001,6.2+0.0001);
			\draw[->, line width=0.2mm, color=black] (2.2,4.6) -- (2.2,6) node[pos=0.8, anchor = west] {$\nfoldx_0$};
\begin{scope}[shift={(3.8,0)}]
			\draw[xstep=0.2cm, ystep=0.2cm,black,thin] (1.0-0.0001,1.4-0.0001) grid (3.0+0.0001,1.6+0.0001);
   		\node[align=left, anchor=south west] at (1.0,0.8) {$m_{t-1}\in\lprodmset$};
     	\draw[xstep=0.2cm, ystep=0.2cm,black,thin] (2.6-0.0001,2.2-0.0001) grid (3.4+0.0001,2.4+0.0001);
   		\node[align=left, anchor=south west] at (2.6,1.7) {$r_{t-1}{\in_R}\mathcal{R}$};
          \draw[xstep=0.2cm, ystep=0.2cm,black,thin] (4.8-0.0001,3.0-0.0001) grid (5.6+0.0001,3.2+0.0001);
   		\node[align=left, anchor=south west] at (4.8,2.5) {$s{\in_R}\mathcal{S}$};
     			\draw[xstep=0.2cm, ystep=0.2cm,black,thin] (1.0-0.0001,4.4-0.0001) grid (4.0+0.0001,4.6+0.0001);
			\draw[->, line width=0.2mm, color=black] (2,1.6) -- (2,4.4) node[pos=0.9, anchor = west] {$\kfoldv_{t-1}$};
   			\draw[->, line width=0.2mm, color=black] (2.6,2.3) -- (2.0,2.3);
         			\draw[->, line width=0.2mm, color=black] (4.8,3.1) -- (2.0,3.1);
     			\draw[xstep=0.2cm, ystep=0.2cm,black,thin] (1.0-0.0001,6.0-0.0001) grid (4.6+0.0001,6.2+0.0001);
			\draw[->, line width=0.2mm, color=black] (2.2,4.6) -- (2.2,6) node[pos=0.8, anchor = west] {$\nfoldx_{t_1}$};
   			\draw[->, line width=0.2mm, color=black] (5.2,3.2) -- (5.2,6) node[pos=0.85, anchor = west] {$\mathrm{ECC}(s)$};
          \draw[xstep=0.2cm, ystep=0.2cm,black,thin] (4.8-0.0001,6.0-0.0001) grid (5.8+0.0001,6.2+0.0001);
\end{scope}
        \end{tikzpicture}
       \caption{\label{fig:wiretap-scheme}The protocol on the sender's side of the seeded wiretap scheme (left) and the unseeded wiretap scheme (right). The seed and additional randomness are chosen uniformly at random.}
\end{figure}
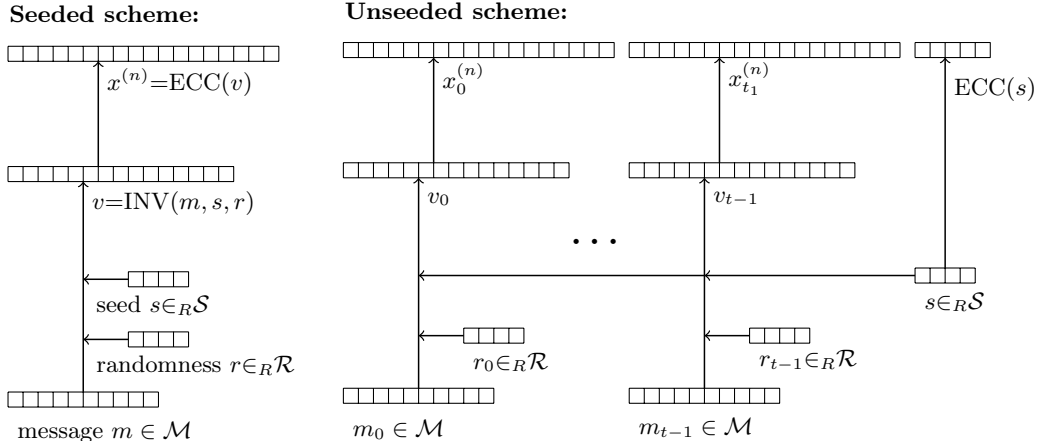

Security and correctness of this scheme are analyzed in~\cite{tessaroarxiv} where the following properties are shown. We explicitly include the adversary's strategy here.
Correctness relies \emph{only} on the error-correcting code and the protocol inherits the correctness properties of the code. 
\begin{lemma}
Protocol~\ref{prot:prot} is $\epsilon^{mt}_{cor}$-correct if for any adversarial strategy $w\in\mathcal{W}$
\begin{align*}
\max_{\kfoldv}\Prob^w[\mathrm{DEC}(ChR(\mathrm{ECC}(\kfoldv)))\neq \kfoldv]\leq \epsilon^{mt}_{cor}\ .
\end{align*}
Protocol~\ref{prot:prot} is $\epsilon^{rm}_{cor}$-correct if for any adversarial strategy $w\in \mathcal{W}$
\begin{align*}
\sum_{\kfoldv\in \kfoldV}\frac{1}{\kprodvsetsize} \Prob^w [\mathrm{DEC}(ChR(\mathrm{ECC}(\kfoldv)))\neq \kfoldv]\leq \epsilon^{rm}_{cor}\ .
\end{align*}
\end{lemma}

Since the scheme can be combined with any error-correcting code, we will not focus on correctness, but simply show the secrecy of different setups under the assumption that a reasonably good and efficient error-correcting code is known.

\begin{lemma}\label{lemma:definition_of_secrecy}
For any $\epsilon>0$ and adversarial strategy $w\in \mathcal{W}$, Protocol~\ref{prot:prot} reaches secrecy 
\begin{align}\label{eq:securityformula}
\epsilon_{sec}^{rm}&\leq \frac{1}{2}\sqrt{\lprodmsetsize P^{\varepsilon\ w}_{\mathrm{guess}}(\kfoldV|\nfoldZ)}+\varepsilon\ .
\end{align} 
\end{lemma}
The proof~\cite{tessaroarxiv} uses the fact that choosing the message uniformly at random $m\in_R \lprodmset$ leads to a uniform distribution of $\kfoldv$. It then applies the definition of the min-entropy in terms of the guessing probability and applies the left-over hash lemma for the smooth min-entropy.

Security for arbitrary messages can be related to the security for random messages for certain schemes and channels~\cite{tessaroarxiv}. We give a generalization of this reduction in Section~\ref{sec:distinguishing_security}.

The above scheme uses a \emph{strong extractor} which takes a random seed as second input. Such a seed is necessary in an approach which bases security \emph{only} on the min-entropy
and does not take into account the exact structure of the channel or code. 
Since the extractor is strong, the seed can, however, be reused and it can be leaked to the adversary. 

In the literature, the seed is usually treated in one of two ways: it is considered to be previously fixed and known to all parties~\cite{hayashimatsumoto}; e.g., it could be chosen and hardcoded once and for all upon manufacturing of the communication devices. Alternatively, it can be communicated over the communication channel~\cite{tessaroarxiv}. Since the same seed can be used for several message blocks, this does not affect the asymptotic rate. 

Note, however, that the seed is required to be independent from the information the adversary receives, i.e., the random variable $\nfoldZ$. In situations where the eavesdropper cannot influence or change the wiretap channel, this is always fulfilled. Here, on the other hand, we  allow the adversary to choose the channel within some limits. To ensure the independence from the seed value, the seed could be sent over the channel \emph{last}, i.e., after all messages have been sent. This would ensure that the adversary cannot adapt the channel properties depending on the observed seed value. 

\begin{protocol}[Unseeded Wiretap scheme]\label{prot:unseeded} \phantom{a}
\begin{enumerate}
\item The sender chooses a seed $s$ uniformly at random, $s\in_R\mathcal{S}$.
\item Repeat $t$ times the seeded wiretap scheme:
\begin{enumerate}
\item \label{unseed:enum:1} The sender chooses a message $\lfoldm\in \lprodmset$ and randomness $r\in_R \mathcal{R}$.
\item \label{unseed:enum:2} The sender calculates $\nfoldx{=}\mathrm{ECC}(\mathrm{INV}(\lfoldm,s,r))$.
\item \label{unseed:enum:3} The sender sends $\nfoldx$ over the channel $ChR^{(n)}$.
\item  \label{unseed:enum:4} The receiver obtains $\nfoldy$ with  $\nfoldY{=}ChR^{(n)}(\nfoldx)$.
\end{enumerate}
\item The sender calculates $a{=}\mathrm{ECC}(s)$ and sends it over the channel.
\item The receiver obtains $a'$ according to $ChR(a)$ and calculates $s'{=}\mathrm{DEC}(a')$.
\item \label{unseed:enum:5} For all $t$ values $\nfoldy$, the receiver calculates ${\lfoldm}'{=}\mathrm{EXT}(\mathrm{DEC}(\nfoldy),s')$.
\end{enumerate}
\end{protocol}

Both correctness and secrecy `behave well' under this composition~\cite{tessaroarxiv} as stated by the following lemmas. They remain valid when including the adversary's strategy because of the convexity of maximizing. 
\begin{lemma}
Let the (seeded) Protocol~\ref{prot:prot} be $\epsilon_{cor}^{mt}$-correct for arbitrary messages and $\epsilon_{cor}^{rm}$-correct for random messages. Then the (unseeded) Protocol~\ref{prot:unseeded} is correct with
\begin{align*}
    \epsilon_{cor}^{mt\ \mathrm{unseeded}}&\leq  t\cdot \epsilon_{cor}^{mt} + \max_{w\in \mathcal{W}}\sum_{s\in \mathcal{S}}\frac{1}{\abs{\mathcal{S}}} \Prob\Big[\mathrm{DEC}(ChR(\mathrm{ECC}(s)))\neq s\Big]\\
    \epsilon_{cor}^{rm\ \mathrm{unseeded}}&\leq  t\cdot \epsilon_{cor}^{rm} + \max_{w\in \mathcal{W}}\sum_{s\in \mathcal{S}}\frac{1}{\abs{\mathcal{S}}}\Prob\Big[\mathrm{DEC}(ChR(\mathrm{ECC}(s)))\neq s\Big] \ .  
\end{align*}
\end{lemma}
The second term is simply the probability of correctly transmitting the seed, which is always chosen at random.

Since the seed can be public, secrecy of the unseeded wiretap protocol is bounded by $t$ times the secrecy of the seeded protocols.
\begin{lemma}
Let the (seeded) Protocol~\ref{prot:prot} be $\epsilon_{sec}^{rm}$-secure for random messages.
Then the (unseeded) Protocol~\ref{prot:unseeded} is $\epsilon_{sec}^{rm\ \mathrm{unseeded}}$-secret with
\begin{align*}
     \epsilon_{sec}^{rm\ \mathrm{unseeded}}&\leq  t\cdot \epsilon_{sec}^{rm} \ . 
\end{align*}
\end{lemma} 

The above argument allows us in the following Section~\ref{sec:result} to focus on the seeded wiretap channel with random message input.

\section{Security Bound}\label{sec:result}

This section contains our main result, giving a simple but effective way to bound random-message security for the wiretap channel. Our approach is inherently agnostic to the \emph{order} in which the adversarial channels are applied and therefore holds even if the adversary can choose this. 

To prove security, we will proceed in steps. We first (Section~\ref{sec:simplest_bound}) give a very simple security bound by directly applying the definition of the min-entropy of a random variable as the guessing probability. As an example, Section~\ref{sec:ex_wiretap2} shows that this allows to reach secrecy capacity for the wiretap channel II.  We then refine our method by introducing `smoothing' into our analysis (Section~\ref{sec:smoothing}), which allows us to reach capacity for a large(r) class of channels, including the binary symmetric channel, as shown as example in Section~\ref{sec:ex_bsc_capacity}. 
In Section~\ref{sec:general_bound}, we give a general formula for the security which bases on the asymptotic equipartition property.  The channel uses need to be neither the same (in terms of transition probabilities) nor independent. As an example, we show security for a special type of arbitrarily varying wiretap channel (Section~\ref{sec:ex_arbitrarily_varying}), where the different channel types have a fixed frequency, but the order can be chosen by the adversary.

\subsection{Simple Security Bound}\label{sec:simplest_bound}

To show security, by Lemma~\ref{lemma:definition_of_secrecy}, the main task is to bound the min-entropy, i.e., guessing probability, of $\kfoldV$ given $\nfoldZ$ for random input $\kfoldV$ and any adversarial strategy $w\in \mathcal{W}$
\begin{align}
\nonumber	P^w_\mathrm{guess}(\kfoldV|\nfoldZ)
 &:=\sum_{\nfoldz\in \nprodzset}\max_{\kfoldv\in \kprodvset} 
 \Pr^w [\kfoldV=\kfoldv \land \nfoldZ=\nfoldz]\ .
\end{align}
The idea is that we can simply take the $\nprodzsetsize$ highest probabilities of $\Pr^w [\kfoldV{=}\kfoldv \land \nfoldZ{=}\nfoldz]$. (Recall that, when we write $p_i$, the probabilities are ranked in descending order.)
\begin{lemma}
The guessing probability of $\kfoldV$ given $\nfoldZ$ for random input $\kfoldV$ and an adversarial strategy $w\in \mathcal{W}$ is bounded by 
\begin{align}
P^w_\mathrm{guess}(\kfoldV|\nfoldZ)
\label{eq:highest} &\leq \sum_{i=0}^{\nprodzsetsize-1} {p^w}_i (\kfoldV,\nfoldZ) \ .
\end{align}
\end{lemma}
\begin{proof}
\begin{align}
\nonumber P^w_\mathrm{guess}(\kfoldV|\nfoldZ)&=\sum_{\nfoldz\in \nprodzset}\max_{\kfoldv\in \kprodvset} 
 \Pr^w [\kfoldV=\kfoldv \land \nfoldZ=\nfoldz]\\
 \nonumber &=
 \sum_{\nfoldz\in \nprodzset}
 p^w_0(\kfoldV,\nfoldZ=\nfoldz)
 \leq 
 \sum_{i=0}^{\nprodzsetsize-1} p^w_i (\kfoldV,\nfoldZ)\ .
\end{align}    
\end{proof}

Under the condition that all inputs to the (adversarial) channel lead to exactly the same output distribution upon relabelling of the output symbols we can further simplify this condition and express it in terms of the probabilities given any specific input. The procedure is illustrated in Figure~\ref{fig:example-bound-guessing-prob}. 
\begin{lemma}\label{lemma:guessing_for_uniform_and_equal_prob}
The guessing probability of $\kfoldV$ given $\nfoldZ$ for uniformly random input $\kfoldV$, adversarial strategy $w\in  \mathcal{W}$, and for channels $ChA^{(n)}\circ ECC:\kprodvset\rightarrow \nprodzset$ such that $\vv{p^w}(\nfoldZ|\kfoldV{=}\kfoldv)=\vv{p^w}(\nfoldZ|\kfoldV{=}\kfoldspezv)$ for all $\kfoldv$ is bounded by 
\begin{align*}
P^w_\mathrm{guess}(\kfoldV|\nfoldZ) &\leq \sum_{i=0}^{\round-1} {p^w}_i ( \nfoldZ|\kfoldV=\kfoldspezv)\ .
\end{align*}
\end{lemma}
\begin{proof}
\begin{align*}
P^w_\mathrm{guess}(\kfoldV|\nfoldZ)
&\leq \sum_{i=0}^{\nprodzsetsize-1} p^w_i (\kfoldV|\nfoldZ)\nonumber \leq \sum_{i=0}^{\round -1 } \kprodvsetsize\cdot  p^w_i (\kfoldV=\kfoldspezv,\nfoldZ) \\
&\leq \sum_{i=0}^{\lceil \nprodzsetsize/\kprodvsetsize\rceil -1} p^w_i (\nfoldZ|\kfoldV=\kfoldspezv)\ ,
\end{align*}    
since with uniform inputs $\kfoldV$,  $ p^w_i (\kfoldV{=}\kfoldspezv,\nfoldZ){=} P_{\kfoldV}(\kfoldV{=}\kfoldspezv) p^w_i (\nfoldZ|\kfoldV{=}\kfoldspezv) {=} \frac{1}{\kprodvsetsize} p^w_i (\nfoldZ|\kfoldV{=}\kfoldspezv)$.
\end{proof}

\begin{figure}[ht]
    \centering
    \includegraphics[width=\textwidth]{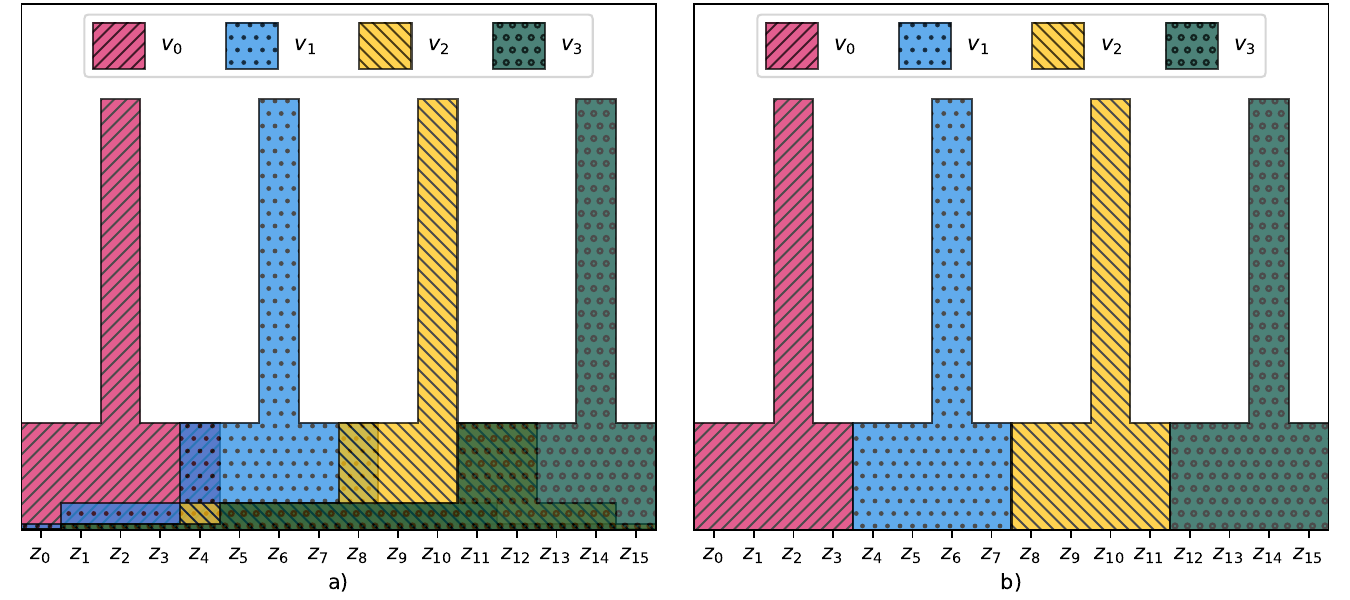}
    \caption{The largest probabilities allow to bound the guessing probability. The example shows the probability distribution of a binary symmetric channel with crossover probability $p_A=0.2$ with $4$ values of $\kfoldV$ (denoted by different colors) and $16$ output values $\nfoldZ$. The histogram on the left depicts the output probabilities from different inputs, the one on the right shows the selected $|\mathcal{Z}|^n$ highest probabilities, $4=16/4$ from each input $\kfoldV$, which allow to bound the guessing probability.}
    \label{fig:example-bound-guessing-prob}
\end{figure}

The bound on the guessing probability allows to bound the distance from uniform from the adversaries' point of view and therefore the secrecy of the scheme.
\begin{lemma}\label{lemma:random_message_security_in_terms_of_guessing_probability} 
Consider a channel $ChA^{(n)}\circ \mathrm{ECC}:\kprodvset\rightarrow \nprodzset$ such that $\vv{p^w}(\nfoldZ|\kfoldV=\kfoldv)=\vv{p^w}(\nfoldZ|\kfoldV=\kfoldspezv)$ for all $\kfoldv$ and $w\in \mathcal{W}$. Then Protocol~\ref{prot:prot} achieves 
\begin{align*}
\epsilon_{sec}^{rm}&\leq \frac{1}{2}\max_{w\in \mathcal{W}}\sqrt{\lprodmsetsize\cdot\sum_{i=0}^{\round -1} p^w_i (\nfoldZ|\kfoldV=\kfoldspezv)}\ .
\end{align*}
\end{lemma}
\begin{proof} By Lemma~\ref{lemma:definition_of_secrecy}, 
\begin{align*}
\epsilon_{sec}^{rm}&\leq \frac{1}{2} \max_{w\in \mathcal{W}}
\sqrt{\lprodmsetsize\cdot {P^w}_{guess}(\kfoldV|\nfoldZ)} \leq \frac{1}{2}\max_{w\in \mathcal{W}}\sqrt{\lprodmsetsize\cdot\sum_{i=0}^{\round -1} p^w_i (\nfoldZ|\kfoldV=\kfoldspezv)}\ .
\end{align*}
\end{proof}

This straight-forward approach to bound the security of the scheme is able to reach secrecy capacity for the Wiretap Channel II~\cite{wiretap2}, as we see in Section~\ref{sec:ex_wiretap2}. While we have to further refine the approach to reach capacity, e.g., for the binary symmetric channel (which we do in Section~\ref{sec:smoothing}), the approach already reaches a good security bound in the non-asymptotic setting with short finite-length messages (Figure~\ref{fig:comparison-bounds}).

\begin{mdframed}
\subsection{Example: Adversarially Selected Set of Bits (Wiretap II)}\label{sec:ex_wiretap2}
Let us consider a simple example where the channels $ChA_{i}$ are not equal and their type may be selected adversarially: the channel where the adversary can select to receive a certain number of bits of the codeword (but not all), e.g. $q$ out of the $n$ symbols~\cite{wiretap2}. For all other symbols, the adversary receives a random symbol. This corresponds to selecting between an error-free binary symmetric channel and a completely noisy symmetric channel. The adversary's input $W$ in Figure~\ref{fig:wiretap} is therefore a bit-string with $q$ $1$'s denoting the selected bits, $w\in \mathcal{W}{=}\{\{0,1\}^n|w_H{=}
q\}$, where $w_H$ denotes the Hamming weight. Note that the adversary can also select to receive fewer symbols, however, this simply corresponds to `forgetting' some of the received ones and the bound still holds. All the considered inputs and outputs are bit-strings, more precisely $\lprodmset{=}\{0,1\}^\ell$, $\kprodvset{=}\{0,1\}^k$ and $\nprodyset{=}\nprodzset{=}\{0,1\}^n$.

The output probability distribution for any specific $w\in \mathcal{W}$ is
\begin{align*}
\Pr^w [ \nfoldZ=\nfoldz|\kfoldV=\kfoldspezv] &= \delta(\nfoldz_{{q}}, ECC(\kfoldspezv)_{{q}})\cdot 2^{q-n}
\end{align*}
where we denoted by $z_{{q}}$ the values the adversary chose to receive and $\delta(\nfoldz_{{q}},ECC(\kfoldspezv)_{{q}})$ is the Kronecker delta. This means, for all $\kfoldv$, 
\begin{align*}
\vv{p^w}_i (\nfoldZ|\kfoldV=\kfoldv)&= 
\begin{cases}
2^{-(n-q)} & \text{for }i=0,\ldots 2^{n-q}-1\\
0 & \text{for }i= 2^{n-q},\ldots 2^n
\end{cases}   
\end{align*}
The guessing probability for any $w\in \mathcal{W}$ is bounded by Lemma~\ref{lemma:guessing_for_uniform_and_equal_prob}
\begin{align*}
	P^w_\mathrm{guess}(\kfoldV|\nfoldZ)
 &\leq 
2^{n-k}\cdot 2^{q-n} = 2^{q-k}\ ,
\end{align*}
and Protocol~\ref{prot:prot} is  $\epsilon_{sec}^{rm}$-secret for the wiretap channel II with
\begin{align*}
\epsilon_{sec}^{rm}=\frac{1}{2} \max_{w\in \mathcal{W}} \sqrt{2^{\ell}\cdot {P^w}_\mathrm{guess}(\kfoldV|\nfoldZ)}
&=
\frac{1}{2} \sqrt{2^{\ell}\cdot 2^{q-k}}\ .
\end{align*}
When the error-correcting code reaches capacity on the receiver channel, i.e. $k\approx n\cdot(1- h(p_R))$, and when $q$ is some fixed frequency $f=q/n$ this allows to reach secrecy capacity~\cite{wiretap2}, i.e., the achieved rate is approximately 
\begin{align*}
    r&\approx 1-h(p_R)-f\ .
\end{align*} 
\end{mdframed}

\subsection{Smoothing}\label{sec:smoothing}

The above calculation only reaches capacity for some types of channels. To improve the bound and to reach capacity for a larger class of channel types, including the binary symmetric channel, we will replace the calculation of the min-entropy by the $\varepsilon$-smooth version of it (see e.g.~\cite{Cachin97a}). For a specific probability distribution, we can simply `cut' the largest probability to obtain an $\varepsilon$-close version of it. 

\begin{figure}[ht]
    \centering
    \includegraphics[width=\textwidth]{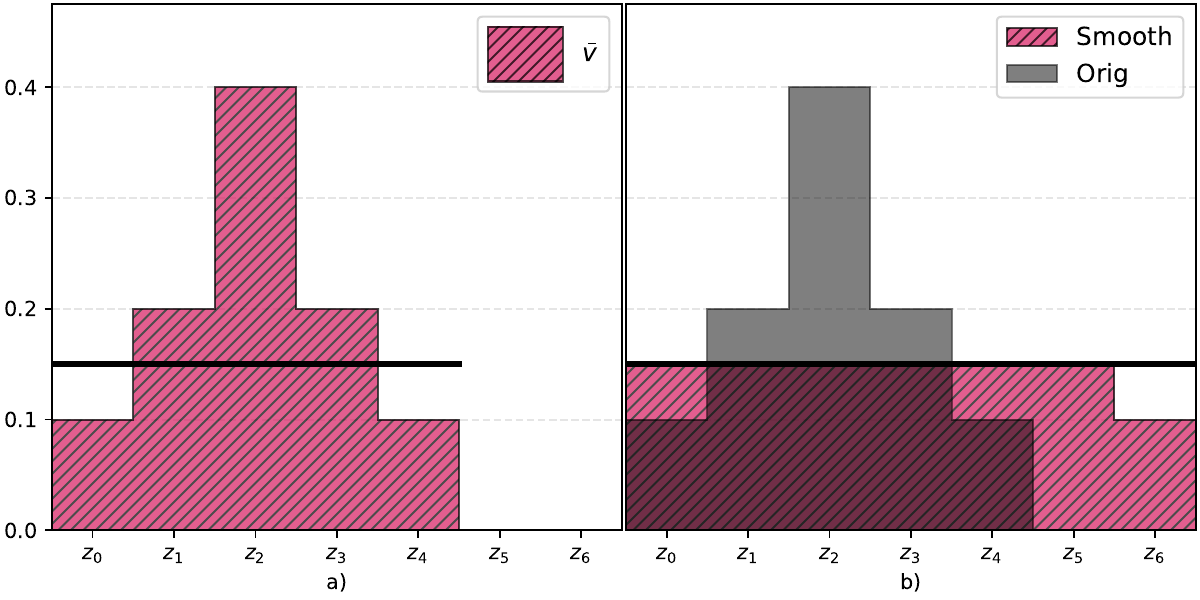}
    \caption{ 
    In a) the output distribution for a specific input $\kfoldspezv$ is shown. The black line at $p=0.15$ shows a possible cutting line to smooth the distribution. In b) an $\epsilon$-smooth version of the same distribution is shown. The area above the cutting line is redistributed across symbols with lower probabilities.}
    \label{fig:example-smoothing}
\end{figure}

The following lemma states that we do not need to `smooth' the complete joint probability distribution, but we can focus on the conditional distribution given one specific input. 
\begin{lemma}\label{lemma:varepsilon-guessing}
The $\varepsilon$-guessing probability of $\kfoldV$ given $\nfoldZ$ for random input $\kfoldV$, adversarial strategy $w\in \mathcal{W}$ and for channels $ChA^{(n)}\circ \mathrm{ECC}:\kprodvset \rightarrow \nprodzset$ such that $\vv{p^w}(\nfoldZ|\kfoldV{=}\kfoldv){=}\vv{p^w}(\nfoldZ|\kfoldV{=}\kfoldspezv)$ for all $\kfoldv$ is bounded by 
\begin{align}\label{eq:pgeps}
P^{\varepsilon\ w}_\mathrm{guess}(\nfoldZ|\kfoldV) &\leq \sum_{i=0}^{\round -1} {\tilde{p}^w}_i (\nfoldZ|\kfoldV=\kfoldspezv)\ ,
\end{align}
where $\tilde{P}^w_{\nfoldZ|\kfoldV=\kfoldspezv}\in \mathcal{P}^{\varepsilon}(P^w_{\nfoldZ|\kfoldV=\kfoldspezv})$
\end{lemma}
\begin{proof}
Since the inputs are uniform and all inputs lead to the same output probability distribution, the distance of the joint distribution $P^w_{\kfoldV \nfoldZ}(\kfoldv,\nfoldz)$ is bound by the distance of the conditional distribution given a certain input. 
\end{proof}

If a distribution only has a small probability to reach an outcome with high associated probability (i.e., the probability has a small peak), then this probability can be `cut' as stated in the following lemma and illustrated in Figure~\ref{fig:example-smoothing}. 
\begin{lemma}\label{lemma:eps_close_distr_exists}
Let $P_\nfoldZ(\nfoldz)$ be a probability distribution over $\nprodzset$ such that
$\sum_{\nfoldz:P_\nfoldZ(\nfoldz)>p}P_\nfoldZ(\nfoldz){\leq}\varepsilon$ for some $p{\geq}1/\nprodzsetsize$. Then, there exists a probability distribution $Q_\nfoldZ(\nfoldz)$ over $\nprodzset$ which is $\varepsilon$-close to $P_\nfoldZ(\nfoldz)$ such that 
\begin{align*}
\sum_{\nfoldz:Q_\nfoldZ(\nfoldz)>p}Q_\nfoldZ(\nfoldz)=0.   
\end{align*}
\end{lemma}
To achieve this, simply `cut' the large probabilities at $p$ and `distribute' the probability above the cutting line among the values with lower probability, as depicted in Figure~\ref{fig:example-smoothing}. 
\begin{proof}
Let $\vv{p}(\nfoldZ)$ be the probability vector associated with the probability distribution $P_\nfoldZ(\nfoldz)$. Define indices $k_1\leq k_2$ such that
\begin{align*}
    k_1 &= \max_j \Big\{j:p_j(\nfoldZ)>p\Big\}\,,\\
    k_2 &= \min_j\Big\{j:\sum_{i=k_1+1}^j\big(p-p_i(\nfoldZ)\big) \geq \sum_{i=0}^{k_1} \big(p_i(Z)-p\big) \Big\}\,.
\end{align*}
Define $Q_\nfoldZ(\nfoldz)$ by it's associated probability vector $\vv{q}(\nfoldZ)$
\begin{align*}
{q}_i(\nfoldZ)&= 
\begin{cases}
p & \text{for }i=0,\ldots,k_2-1\\
p-\eta & \text{for }i=k_2\\
{p_i}(\nfoldZ) & \text{for all remaining }i
\end{cases}   
\end{align*}
where $\eta{=}\sum\limits_{i=k_1+1}^{k_2}\big(p-{p_i}(\nfoldZ)\big)-\sum\limits_{i=0}^{k_1}\big({p_i}(\nfoldZ)-p\big)$.
By construction, the distance between $Q_\nfoldZ(\nfoldz)$ and $P_\nfoldZ(\nfoldZ)$ is 
\begin{align*}
    d(P,Q)&\leq \frac{1}{2}\sum_i \abs{{p_i}(\nfoldZ)-{q_i}(\nfoldZ)}=\sum_{i=0}^{k_2}({p_i}(\nfoldZ)-p)\leq \varepsilon\ .
\end{align*}
\end{proof}

\begin{mdframed}
\subsection{Example: Reaching Capacity for the Binary Symmetric Channel}\label{sec:ex_bsc_capacity}

Take as an example the binary symmetric channel\footnote{This example does not contain an adversarial input $w$.} with 
$\lprodmset{=}\{0,1\}^\ell$, $\kprodvset{=}\{0,1\}^k$ and $\mathcal{X}=\mathcal{Y}=\mathcal{Z}=\{0,1\}$. The output probabilities correspond to a Bernoulli trial,  
\begin{align*}
P_{\nfoldZ|\kfoldV}(\nfoldz,\kfoldspezv)=(1-p_A)^{n-d_H(\nfoldx,\nfoldz)}p_A^{d_H(\nfoldx),\nfoldz)}\ , 
\end{align*}
where $\nfoldx{=}\mathrm{ECC}(\kfoldspezv)$ is the codeword obtained from $\kfoldspezv$ and $d_H$ denotes the Hamming distance. W.l.o.g.\ we can assume that $p_A<1/2$ and the highest terms will be the ones with the lowest Hamming distance.

To `smooth' this distribution, we bound the sum of the highest $2^{n-k}$ probabilities of a Bernoulli trial. We `cap' all probabilities at $\tilde{p}=(2^{-h(p_A-\delta)})^n$. The $\varepsilon$ describing the distance to the original distribution then becomes 
\begin{align*}
\varepsilon&= \sum_{j=0}^t\binom{n}{j}[ p_A^j(1-p_A)^{n-j}-\tilde{p}]\\
&\leq \sum_{j=0}^t\binom{n}{j} p_A^j(1-p_A)^{n-j}\leq e^{-2n(p_A-t/n)^2}=e^{-2n\delta^2}\ ,
\end{align*}
where we have used a Chernoff bound and where $t$ is chosen such that $\sum_{j=0}^t\binom{n}{j}=2^{n-k}$. 
The $\varepsilon$-guessing probability is now given by 
\begin{align}
\nonumber
\tilde{P}^{\varepsilon}_\mathrm{guess}(\kfoldV|\nfoldZ)&\leq 
 \sum_{i=0}^{q\cdot n}\binom{n}{i}\tilde{p}=2^{n-k}\tilde{p}\ ,
\end{align}
and we obtain the security bound
\begin{align}
\nonumber
\epsilon_{sec}^{rm}& \leq
\frac{1}{2} \sqrt{2^{\ell}\cdot 2^{n-k}\cdot 2^{-n h(p_A-\delta)}}+
e^{-2n\delta^2}
\end{align}
for any value of $\delta$. If capacity can be reached on the receiver channel, i.e., $k \approx n\cdot (1- h(p_R))$, this bound vanishes for any $\ell<n\cdot\big( h(p_R)-h(p_A)\big)$ and an appropriately chosen $\delta$, therefore reaching secrecy capacity.
\end{mdframed}

\subsection{A General Bound}\label{sec:general_bound}

We can now apply this method of `cutting' the highest probabilities to any channel with the property that  the probability to obtain a high probability outcome is low. This is, in particular, the case when an AEP holds for the distribution. The following lemma can be seen as our main result allowing to bound random-message security for any channel with the same output distribution for all inputs and where an AEP holds for this conditional distribution.

\begin{lemma}\label{lemma:general_security_bound}
Let $ChA^{(n)}\circ \mathrm{ECC}:\kfoldV\rightarrow \nfoldZ$ be a channel such that for all adversarial strategies $w\in\mathcal{W}$ and all inputs $\kfoldV{=}\kfoldv$ the output distribution is the same upon relabelling of the symbols, i.e., fix any $\kfoldspezv \in \kfoldV$. Suppose for all $\kfoldv \in \kfoldV$,
$\vv{p^w}(\nfoldZ|\kfoldV{=}\kfoldv){=}\vv{p^w}(\nfoldZ|\kfoldV{=}\kfoldspezv)$ and for some $\kappa$ and $\varepsilon$,
\begin{align}
\label{eq:limit_high_prob} \sum_{\nfoldz: P^w_{\nfoldZ|\kfoldV=\kfoldspezv}(\nfoldz)>\kappa}{P^w_{\nfoldZ|\kfoldV=\kfoldspezv}}(\nfoldz)&\leq \varepsilon\ .
\end{align}
Then 
\begin{align*}
\epsilon_{sec}^{rm}&\leq \frac{1}{2}\sqrt{\lprodmsetsize\frac{\nprodzsetsize}{\kprodvsetsize} \kappa}+\varepsilon\ .
\end{align*}
\end{lemma}
\begin{proof}
By Lemma~\ref{lemma:varepsilon-guessing} the $\varepsilon$-guessing probability can be bounded as 
\begin{align*}
P^{\varepsilon\ w}_\mathrm{guess}(\kfoldV|\nfoldZ) &\leq \frac{\nprodzsetsize}{\kprodvsetsize} \kappa\ ,
\end{align*}
with $\varepsilon$ the error and $\kappa$ `cut-off level'  as in~\eqref{eq:limit_high_prob}. The security bound then follows by~\eqref{eq:securityformula}.
\end{proof}
For the case of a distribution which follows an AEP for any $w$ this bound amounts to 
\begin{align*}
\epsilon_{sec}^{rm}&\leq 
   \frac{1}{2}\sqrt{\lprodmsetsize\frac{\nprodzsetsize}{\kprodvsetsize} 2^{-n\big(\mathrm{H}(Z|X=\bar{x})-\delta\big)}}+ \frac{\mathrm{Var}[-\log_2 P(Z|\kfoldV=\kfoldspezv)]}{n\delta^2}\ ,
\end{align*}
where $\mathrm{H}(Z|\kfoldV{=}\kfoldspezv)$ and $\mathrm{Var}[-\log_2 P(\kfoldV)]$ are the values associates with an individual channel use. We are therefore able to show security whenever an AEP holds.

Note that we have not used that the individual channel uses need to be identical or independent, a bound obtained this way therefore applies to any order of the channels and this order can even be chosen by the adversary. Indeed, the adversary can even pick the \emph{exact} channel from any set of channels for which this AEP holds.  

\begin{figure}[ht]
    \centering
    \includegraphics[width=\textwidth]{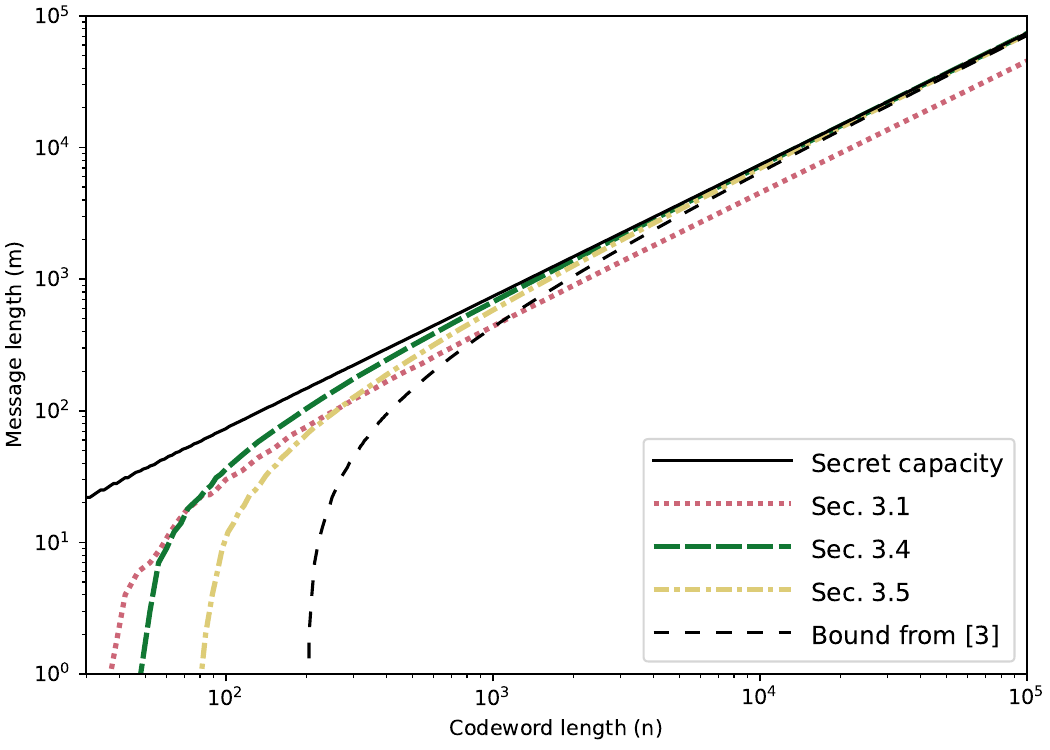}
    \caption{Secret message length vs codeword length using the bound from Section~\ref{sec:simplest_bound}, from Section~\ref{sec:ex_bsc_capacity} and Section~\ref{sec:general_bound}. For comparison, the secrecy capacity and the bound from~\cite{tessaroarxiv} are given.  The parameters used are $p_r=0.03$, $p_a=0.35$ and $\epsilon_\text{sec}\leq10^{-2}$. We state the bound for random messages, the bound for arbitrary message distributions is the above multiplied by a factor of $2$.}
    \label{fig:comparison-bounds}
\end{figure}

\begin{mdframed}
\subsection{Example: Security of a Type-Constrained Arbitrarily Varying Wiretap Channel}\label{sec:ex_arbitrarily_varying}

Consider the case where the $n$ different channel uses are memoryless, but not necessarily identical. The channel acting on each input is affected by a \emph{state} which only influences the current channel use. The state is denoted by $q_i$, i.e., each channel to the adversary is described by the transition matrix $W(z_i|x_i,q_i):=P_{Z_i|X_i,Q_i=q_i}(z_i,x_i,q_i)$. When the channel input $\nfoldx$ is a sequence of $n$ symbols, the probability of outcome $\nfoldz$ is therefore given by $P_{\nfoldZ|\nfoldX,\nfoldQ=\nfoldq}(\nfoldz,\nfoldx,\nfoldq)=\prod_i W(z_i|x_i,q_i)$. Assume that the adversary can choose the sequence of the states subject to the condition that the frequency of every possible state is predetermined, i.e., for every possible state $q$, the total proportion of positions where $q_i=q$ is fixed in advance; the adversary can however choose the order of the states.

It is easy to check that \eqref{eq:pgeps} still holds for such a channel and that the bound obtained from~\eqref{eq:pgeps} remains the same no matter the order of the state. Furthermore, as we shall see in Section~\ref{sec:achievable_rate} (see also Theorem~\ref{th:securityforavwtc}), this bound can achieve secrecy capacity when capacity of each individual channel is achieved by a uniform input distribution on that channel (this is the case, e.g., for strongly symmetric channels) and under the condition that an error-correcting code reaching Shannon capacity for the receiver's channel is known for the same distribution.
\end{mdframed}

\section{Achievable Rate}\label{sec:achievable_rate}

We now show that our approach to bound security is able to reach capacity with Protocol~\ref{prot:prot} for many channels. 

\begin{condition}\label{condition:for_capacity}We impose the following conditions on the channel:
\begin{enumerate}[label=(\alph*)]
    \item\label{item:cond_same_output_distribution} The channel between sender and eavesdropper is such that for any strategy $w\in \mathcal{W}$ all inputs lead to the same output distribution, upon relabelling of the values, i.e., fix any $\kfoldspezv \in \kfoldV$, then $\vv{p^w}(\nfoldZ|\kfoldV{=}\kfoldv){=}\vv{p^w}(\nfoldZ|\kfoldV{=}\kfoldspezv)$ for all $\kfoldv\in \kfoldV$ .
    \item\label{item:cond_capacity_for_receiver_channel} There is a capacity-achieving code for uniform $\kfoldV$ and for any strategy $w\in \mathcal{W}$ for the channel between sender and receiver, i.e., $\kprodvsetsize \rightarrow 2^{n C_R}$ and $\epsilon_{corr}\rightarrow 0$ as $n\rightarrow \infty$ with this code. 
    \item\label{item:cond_receiver_aep} For a fixed input and for any strategy $w\in \mathcal{W}$, the eavesdropper's output distribution $P^w_{\nfoldZ|\kfoldV{=}\kfoldv}(\nfoldz)$ follows an AEP, i.e.,
  \begin{align*}
\Pr[P^w_{\nfoldZ|\kfoldV{=}\kfoldv}(z_0 \ldots z_n)> 2^{-n\big(\mathrm{H}(Z|\kfoldV=\kfoldv)-\epsilon\big)}] &\leq \delta_n\ ,
\end{align*}
where $\delta_n{=} \frac{\mathrm{Var}[-\log_2 P(Z|\kfoldV{=}\kfoldv)]}{n\epsilon^2}$ and we require $\mathrm{Var}[-\log_2 P(Z|\kfoldV=\kfoldv)]$ to be bounded by a constant. 
\end{enumerate}
\end{condition}

\begin{lemma}\label{lemma:capacity_reaching}
    Under Condition~\ref{condition:for_capacity}, the scheme reaches an asymptotic secure message length of 
    \begin{align*}
        \ell&= \log_2 \lprodmsetsize= n \left( C_R - \log_2|\mathcal{Z}| + \mathrm{H}(Z|X)\right)\ ;
    \end{align*}
    i.e., an asymptotic rate of $C_R - \log_2|\mathcal{Z}| + \mathrm{H}(Z|X)$, where $C_R$ denotes the Shannon capacity of $ChR$.
\end{lemma}
\begin{proof}
We have to show that both the correctness and the security parameter tend to zero for a large number of channel uses. Correctness is implied by assumption of Condition~\ref{condition:for_capacity}~\ref{item:cond_capacity_for_receiver_channel}. It remains to show that the secrecy parameter vanishes. The secrecy bound is given by
\begin{align*}
\epsilon_{sec}^{rm}&\leq \frac{1}{2}\max_{w\in \mathcal{W}}\sqrt{\lprodmsetsize P^{\delta_n,w}_\mathrm{guess}(\kfoldV|\nfoldZ)}+\delta_n
\leq \frac{1}{2}\sqrt{2^\ell\frac{\nprodzsetsize}{\kprodvsetsize}\kappa}+\delta_n\\
&=
\frac{1}{2}\sqrt{2^{\ell-n\big( \log_2\kprodvsetsize-\log_2|\mathcal{Z}|+\log_2(\kappa)\big)}}+\delta_n\ ,
\end{align*}
where $\kappa$ is the parameter in~\eqref{eq:limit_high_prob}. 
Asymptotically, by Condition~\ref{condition:for_capacity}~\ref{item:cond_capacity_for_receiver_channel} $\log_2\kprodvsetsize \rightarrow n C_R$; by Condition~\ref{condition:for_capacity}~\ref{item:cond_receiver_aep}
$\delta_n\rightarrow 0$ and $\log_2(\kappa)\rightarrow -n \mathrm{H}(Z|\kfoldV{=}\kfoldspezv)$. However, since $P(Z|\kfoldV{=}\kfoldspezv){=}P(Z|X=\mathrm{ECC}(\kfoldspezv))$ we can replace the input $\kfoldV$ by $X$. Furthermore, by Condition~\ref{condition:for_capacity}~\ref{item:cond_same_output_distribution} all output distributions are the same and therefore $\mathrm{H}(Z|X{=}\bar{x})=\mathrm{H}(Z|X)$ for any $\bar{x}\in \mathcal{X}$. 
\end{proof}

If a uniform input yields a uniform output $Z$ at the adversary, then the above equation becomes $C_R - I(X;Z)$, which equals the secrecy capacity. That is, the method achieves secrecy capacity when a uniform input yields a uniform $Z$ at the adversary.
This is, in particular, true when the individual channels are strongly symmetric.

\section{Secrecy for Arbitrary Message Distributions}\label{sec:distinguishing_security}

The ultimate goal is to obtain distinguishing security for arbitrary message distributions. In the following, we give conditions under which secrecy for random messages implies secrecy for arbitrary messages. While this is not the case in general, in~\cite{bellare2012cryptographic} it was shown to hold for  \emph{symmetric} discrete memoryless channels \emph{with binary inputs}. We generalize their proof to the case when the different symmetric channels are memoryless, but not necessarily identical and when the channels operate on non-binary inputs. 

A key ingredient to show the security for arbitrary message distributions is a lemma which states that for certain channels --- more precisely \emph{symmetric} channels --- all specific inputs lead to an output distribution with the same distance from the distribution when the inputs are chosen uniformly. The following is Lemma~5.8 from~\cite{tessaroarxiv} with minimally adapted notation.
\begin{lemma}[Lemma~5.8 from~\cite{tessaroarxiv}]\label{lemma:bound_by_twice_the_distance}
Let $Ch:\lprodmset \rightarrow \nprodzset$ be a symmetric channel. Let $U$ be uniformly distributed over $\lprodmset$. Then there exists a $\Delta$ such that for all $\lfoldm\in \lprodmset$
\begin{align*}
   \Delta&= d(Z(U);Z(\lfoldm))
\end{align*}
and for any $\lfoldm_0,\lfoldm_1\in \lprodmset$
\begin{align*}
   d(Z(\lfoldm_0);Z(\lfoldm_1))&\leq 2\Delta\ ,
\end{align*}
where we denoted by $Z(\lfoldm)$ and $Z(U)$ the output distribution of the adversarial channel upon input message $m$ and upon uniform input, respectively.
\end{lemma}
\begin{proof}
The channel is symmetric, so there exists a partition of the outputs $\mathcal{Z}{=} \bigcup_v \mathcal{Z}_v$ such that each submatrix of $W$ induced by an element of the partition is strongly symmetric. Consider now the distance from uniform when restricting to one such element of the partition. All output values within one element of the partition have the same probability when the input is uniform. Since all rows of the transition matrix are permutations of each other, every input leads to the same distance from uniform. This holds for every element of the partition, therefore it also holds for all outputs. Finally, by  the triangle inequality, the distance between the output distribution of any two messages is bounded by twice the distance from uniform. 
\end{proof}

We have required for the channel $ChA^{(n)}$ to be symmetric. However, the message is not directly input into the channel. The application of the inverter and the error correcting code imply that only a subset of the possible inputs to the channel occur. We now show a condition, under which the `new' channel, restricting to certain inputs, is still symmetric and therefore the distance from uniform (over the restricted subset) is still the same for all specific inputs.
In the following, we restrict the channel inputs to be of the form $\nprodxset=F^n$ with $F$ of the form $\mathbb{Z}/p$ for some prime $p$. 

We first state some simple lemmas which we will use below. Namely, the product of two symmetric channels is symmetric.
\begin{lemma}\label{lemma:product_is_symmetric} 
Let $Ch:\mathcal{X}^2\rightarrow \mathcal{Z}^2$ be a channel such that $Ch=Ch_1\otimes Ch_2$ and each $Ch_i:\mathcal{X}_i\rightarrow \mathcal{Z}_i$ is symmetric. Then $Ch$ is symmetric.
\end{lemma}
\begin{proof}
Since $Ch_1$ and $Ch_2$ are symmetric, there exists a partition of their respective outputs $\mathcal{Z}_1{=} \bigcup_v \mathcal{Z}_{1v}$ and $\mathcal{Z}_2{=} \bigcup_w \mathcal{Z}_{2w}$ such that the channels induced by the partition are strongly symmetric. Partition the outputs of the joint channel according to $(\mathcal{Z}_{1v},\mathcal{Z}_{2w})$. Then 
\begin{align*}
    W(z_1z_2|x_1x_2)&= W(z_1|x_1)\cdot W(z_2|x_2)= W(\pi_1^{x_1\mapsto x'_1}(z_1)|x'_1)\cdot W(\pi_2^{x_2\mapsto x'_2}(z_2)|x'_2)= W(z'_1|x'_1)\cdot W(z'_2|x'_2)\\
     W(z_1z_2|x_1x_2)&= W(z_1|x_1)\cdot W(z_2|x_2)= W(z'_1|\pi_1^{z_1\mapsto z'_1}(x_1))\cdot W(z'_2|\pi_2^{z_2\mapsto z'_2}(x_2))= W(z'_1|x'_1)\cdot W(z'_2|x'_2)   
\end{align*}
for $x_1\neq x'_1$, $x_2\neq x'_2$ and $z_1\neq z'_1\in \mathcal{Z}_{1v}$ and $z_2\neq z'_2\in \mathcal{Z}_{1w}$ and therefore $(z_1,z_2)\neq (z'_1,z'_2)\in (\mathcal{Z}_{1v},\mathcal{Z}_{1w})$. 
\end{proof}

\begin{lemma}\label{lemma:single_input_is_symmetric} 
Let $Ch:\mathcal{X}\rightarrow \mathcal{Z}$ be a channel with only a single input element, i.e., $\|\mathcal{X}\|=1$. This channel is symmetric.
\end{lemma}
\begin{proof}
    Partition the set of output such that each element of $\mathcal{Z}$ is in a separate partition. The matrix induced when restricting to this output set consists of a single entry $W(z|\bar{x})$, which is clearly strongly symmetric.
\end{proof}

Furthermore, the combination of the channel with an additional channel which outputs a symbol that is a deterministic function of the \emph{previous} inputs is also symmetric.

\begin{lemma}\label{lemma:function_of_previous_is_symmetric} 
Let $Ch_1:\mathcal{X}_1\rightarrow \mathcal{Z}_1$ and $Ch_2:\mathcal{X}_2\rightarrow \mathcal{Z}_2$ be symmetric channels. 
Consider the channel $Ch:Ch_1\otimes Ch_2$ restricting to inputs $(x_1,x_2=f(x_1))$ for a surjective function $f$ such that all $x_2\in \mathcal{X}_2$ have the same number of preimages. Then $Ch$ is symmetric.
\end{lemma}
\begin{proof}
Partition the outputs of the joint channel according to $(\mathcal{Z}_{1v},\mathcal{Z}_{2w})$. 
We use the fact that $Ch_1$ and $Ch_2$ are symmetric and therefore 
\begin{align*}
    W(z_1z_2|x_1x_2)&= W(z_1z_2|x_1f(x_1))= W(z_1|x_1)\cdot W(z_2|f(x_1))\\
    &= W(\pi_1^{x_1\mapsto x'_1}(z_1)|x'_1)\cdot W(\pi_2^{f(x_1)\mapsto f(x'_1)}(z_2)|f(x'_1))= W(z'_1|x'_1)\cdot W(z'_2|x'_2)=  W(z'_1z'_2|x'_1x'_2)
\end{align*}
for $x'_2=f(x'_1)$ which is therefore still a valid input. Additionally, by the condition that the two channels are symmetric, when $z_1\in \mathcal{Z}_{1v}$ then  $z'_1\in \mathcal{Z}_{1v}$ and when $z_2\in \mathcal{Z}_{2w}$ then  $z'_2\in \mathcal{Z}_{2w}$, therefore, when $(z_1,z_2)\in (\mathcal{Z}_{1v},\mathcal{Z}_{2w})$ then $(z'_1,z'_2)\in (\mathcal{Z}_{1v},\mathcal{Z}_{2w})$. 
Furthermore, for every 
\begin{align*}
     W(z_1z_2|x_1x_2)&=  W(z_1z_2|x_1f(x_1))= W(z_1|x_1)\cdot W(z_2|f(x_1))= W(z'_1|\pi_1^{z_1\mapsto z'_1}(x_1))\cdot W(z'_2|\pi_2^{z_2\mapsto z'_2}(f(x_1)))\\
     &= W(z'_1|x'_1)\cdot W(z'_2|x'_2)   
\end{align*}
with $x'_2=\pi_2^{z_2\mapsto z'_2}(f(x_1))$. Since for every $f(x_1)$ and $f(x'_1)$ there exists a $z'_2$ such that $z_2,z'_2\in \mathcal{Z}_{2w}$ and $W(z'_2|f(x'_1)=W(z_2|f(x_1)$ we simply pick $z'_2$ such that this holds.
\end{proof}

Using the above lemmas, we can show that a channel with prime field input is still symmetric when restricting the inputs to the elements of a linear error-correcting code.

\begin{lemma}\label{lemma:condition_for_still_symmetric} 
Let $ChA:\nprodxset\rightarrow \nprodzset$ be a channel such that $ChA=\bigotimes_i ChA_i$ and each $ChA_i:\mathcal{X}_i\rightarrow \mathcal{Z}_i$ is symmetric. Let $\nprodxset=F^n$ with $F$ of the form $\mathbb{Z}/p$ for some prime $p$ and $C\subseteq F^n$ a linear error-correcting code on $F^n$. Then the channel when restricting the inputs to elements of $C$, $ChA:C\rightarrow \nprodzset$ is symmetric.
\end{lemma}
\begin{proof}
Proceed inductively over the dimensions of the code. Since the code forms a linear subspace of the vector space $F^n$, the projection of this space onto the first $k$ dimensions also forms a subspace. For the first dimension, the projection of the code onto the first dimension either contains $1$ (the neutral) or $p$ elements. In the former case the channel acting on the first dimension is symmetric because of Lemma~\ref{lemma:single_input_is_symmetric}. In the latter case the restriction to the code as input is symmetric because of the original symmetry condition.\\
Assume now that the channel acting on the first $k-1$ dimensions is symmetric. Compare the projection of the code onto the first $k$ dimensions with the one onto $k-1$ dimensions. This projection on $k$ dimensions contains either the same number of (different) codewords as the projection on $k-1$ dimensions or $p$ times more codewords. If it contains $p$ times more codewords, the new channel is $Ch^{(k-1)}\otimes Ch_{k}$. $Ch^{(k-1)}$ is symmetric by assumption and the product of two symmetric channels is symmetric by Lemma~\ref{lemma:product_is_symmetric}.\\
If it contains the same number of codewords there are two possibilities: either $x_k=0$ (in which case the symbol $x_k$ is trivial and can be ignored) or $x_k=f(x_{(k-1)})$ for an surjective function $f$ and the channel is symmetric by Lemma~\ref{lemma:function_of_previous_is_symmetric}.
\end{proof}

We can now state and prove the main theorem relating random-message security to security for arbitrary message distribution for channels which are symmetric and memoryless (but not necessarily identical), generalizing Theorem 4.12 in~\cite{bellare2012cryptographic}.
\begin{theorem}\label{th:distinguishing_security}
Let $INV:\mathcal{M}\rightarrow \mathcal{V}$ be an inverter of a strong extractor and $ECC:\mathcal{V}\rightarrow \nfoldX$ be a linear error correcting code over $F^n$ with $F=\mathbb{Z}/p$. Let $ChA^{(n)}:\nprodxset\rightarrow\nprodzset$ be an $n$-fold memoryless but not necessarily identical channel, i.e., $ChA^{(n)}{=}\bigotimes_i ChA_i$. 
Let each $ChA_i$ be symmetric. Then 
\begin{align*}
    \epsilon_{sec}^{mt}&\leq 2\epsilon_{sec}^{rm}\ .
\end{align*}
\end{theorem}
\begin{proof}
By the property of the inverter, a uniform distribution over $\mathcal{M}$ leads to a uniform distribution over $\mathcal{V}$. By Lemma~\ref{lemma:condition_for_still_symmetric}, the channel $[ChA^{(n)}\circ ECC]$ is symmetric and the claim follows from Lemma~\ref{lemma:bound_by_twice_the_distance}.
\end{proof}

The above argument implies that we can achieve secrecy for arbitrary message distributions when the channel is a sequence of memoryless but not identical symmetric channels. In combination with Lemma~\ref{lemma:capacity_reaching}, this implies security for the arbitrarily varying wiretap channel with fixed frequency of types of individual symmetric channels, even if the adversary can choose the state sequence. We formally state this in the following theorem.
\begin{theorem}\label{th:securityforavwtc}
Let $\mathrm{INV}$ be an inverter of a two-universal hash function. 
Let $\mathrm{ECC}$ be a linear error-correcting code. Let the channel $ChA^{(n)}=\bigotimes_i ChA_i$ be such that each $ChA_i$ is 
symmetric and  described by the transition matrix $W(z_i|x_i,q_i)$. The frequency of every possible state $q$ is predetermined as $f_q$. The adversary's input $W$ corresponds to the state sequence. Then  
Protocol~\ref{prot:prot} reaches an asymptotic secure message length of 
    \begin{align*}
        \ell&= \log_2 \lprodmsetsize= \log_2\kprodvsetsize- n \left( \sum_q f_q I_q(X;Z)\right)\ .
    \end{align*}
\end{theorem}
\begin{proof}
Combining Lemma~\ref{lemma:capacity_reaching} with Theorem~\ref{th:distinguishing_security} and using the fact that for strongly symmetric channels and uniform input  $\log_2|\mathcal{Z}| - \mathrm{H}(Z|X)=I(X;Z)$.
\end{proof}

The condition on the inverter holds for both the modified Toeplitz hashing and the multiplication with a field element, we can pick either of them. Furthermore, when the error-correcting code reaches Shannon capacity on the receiver channel, then $\log_2\kprodvsetsize\approx nC_R$ and the complete scheme reaches secrecy capacity for arbitrary message distributions.

\section{Conclusion and Outlook}\label{sec:conclusion}

In this paper, we have shown the security of an explicit efficient scheme for wiretap coding based on two-universal hashing combined with any error-correcting code. The security bound applies to finite-length messages and at the same time reaches capacity asymptotically for a large class of channels which may not be memoryless and which can be influenced by the adversary. With the exception of the special case of wiretap channel II, explicit schemes were previously  unknown for channels that are not memoryless, let alone for channels which can be influenced by the adversary.

Our approach uses certain symmetry conditions to simplify the analysis, e.g.\ all inputs lead to the same output distribution on the adversary's side (upon relabelling of the values). It would be interesting to see which meaningful bounds can be obtained when this condition is relaxed. 

We are able to show security for essentially any distribution for which the high probabilities vanish, i.e., some sort of an AEP holds. Since our security bound only considers the total output probability distribution, it does not need the adversarial channel to be identical and the adversary can, in particular, always pick the order in which the channels are applied. 

\section*{Acknowledgements}
We thank Stefano Tessaro for helpful discussions. This work was supported by the Swiss National Science Foundation Practice-to-Science Grant No 199084.

\bibliographystyle{splncs04}
\bibliography{bib}

\end{document}